\documentclass[a4paper,abstracton,11pt]{scrartcl}
\usepackage[utf8]{inputenc}
\usepackage[T1]{fontenc}
\usepackage{xspace}
\usepackage{xcolor}
\usepackage[american]{babel}
\usepackage{amsmath}
\usepackage{amsthm}
\usepackage{thmtools}
\usepackage{graphicx}
\usepackage{float}

\usepackage{mathtools}

\newcommand\Phase[2]{\subsection{Phase \RNum{#1} -- #2}}

\DeclarePairedDelimiter{\ceil}{\lceil}{\rceil}
\DeclarePairedDelimiter{\floor}{\lfloor}{\rfloor}

\makeatletter
\def\xdot{\@ifnextchar.{}{.}}%
\makeatother

\newcommand{\ie}{i.\,e.}

\newcommand{\eg}{e.\,g.}
\newcommand{\Eg}{E.\,g.}

\DeclareMathOperator{\PROBABILITY}{P\scriptstyle{r}}
\newcommand{\Probability}[1]{\ensuremath{ \PROBABILITY\left[#1\right] }}

\DeclareMathOperator{\EXPECTED}{E}
\newcommand{\Expected}[1]{\ensuremath{ \EXPECTED\left[#1\right] }}

\newtheorem{lemma}{Lemma}
\newtheorem{definition}{Definition}
\newtheorem{theorem}{Theorem}

\declaretheoremstyle[
notefont=\bfseries, notebraces={}{},
bodyfont=\normalfont\itshape,
headformat=Lemma\NOTE%
]{nopar}
\declaretheorem[style=nopar]{lemmax}

\newcommand{\RNum}[1]{\uppercase\expandafter{\romannumeral #1\relax}}

\newcommand{\whp}{with high probability\xspace}

\newcommand{\Whp}{With high probability\xspace}

\DeclareMathOperator{\BIGO}{O}
\DeclareMathOperator{\LITTLEO}{o}
\DeclareMathOperator{\BIGOMEGA}{\Omega}
\DeclareMathOperator{\LITTLEOMEGA}{\omega}
\DeclareMathOperator{\BIGTHETA}{\Theta}

\newcommand{\BigO}[1]{\ensuremath{\BIGO\left(#1\right)}}
\newcommand{\LittleO}[1]{\ensuremath{\LITTLEO\left(#1\right)}}
\newcommand{\BigOmega}[1]{\ensuremath{\BIGOMEGA\left(#1\right)}}
\newcommand{\LittleOmega}[1]{\ensuremath{\LITTLEOMEGA\left(#1\right)}}
\newcommand{\BigTheta}[1]{\ensuremath{\BIGTHETA\left(#1\right)}}

\newcommand{\refLemma}[1]{\autoref{#1}}

\newcommand{\seeProofNL}[1]{{%
\unskip\nobreak\hfil\penalty50
\hskip2em\hbox{}\nobreak\hfil\normalfont{#1}%
\parfillskip=0pt \finalhyphendemerits=0 \par}}
\newcommand{\seeProof}[1]{\seeProofNL{(see \autoref{#1})}}

\renewcommand\seeProof[1]{}

\newcommand{\AlgoPhase}[1]{\vspace{0.5ex}\textbf{Phase \RNum{#1}}\\}

\def\ifmonospace{\ifdim\fontdimen3\font=0pt }

\def\Cpp{%
\ifmonospace%
    C++%
\else%
    C\kern-.1667em\raise.30ex\hbox{\smaller{++}}%
\fi%
\spacefactor1000 }

\usepackage{hyperref}
\usepackage{tabularx}
\usepackage{comment}

\usepackage[boxed]{algorithm2e}

\SetAlCapFnt{\footnotesize}
\SetAlCapNameFnt{\footnotesize}
\usepackage{etoolbox}
\makeatletter
\patchcmd{\algocf@latexcaption}{#3}{#3\endgraf}{}{}
\makeatother
\SetAlCapSkip{1ex}
\SetAlFnt{\small}

\hypersetup{colorlinks=true, urlcolor=black, linkcolor=black, citecolor=black}

\definecolor{darkblue}{rgb}{0,0.2,0.4}

\addto\extrasamerican{%
\renewcommand*{\algorithmautorefname}{Algorithm}%
}%

\title{On the Influence of Graph Density on Randomized Gossiping}

\author{Robert Elsässer \and Dominik Kaaser \and \multicolumn{1}{p{0.8\textwidth}}{
\centering \normalsize{\emph{University of Salzburg \\ Department of Computer Sciences \\ 5020 Salzburg, Austria \\ \{elsa, dominik\}@cosy.sbg.ac.at}}}
}

\addtokomafont{disposition}{\rmfamily}

\date{}

\begin{document}
\maketitle
\thispagestyle{empty}
\begin{abstract}

Information dissemination is a fundamental problem in parallel and distributed
computing. In its simplest variant, known as the broadcasting problem, a
single message has to be spread among all nodes of a graph. A prominent communication
protocol for this problem is based on the so-called random phone call model
(Karp et~al., FOCS 2000). In each step, every node opens a communication
channel to a randomly chosen neighbor, which can then be used for
bi-directional communication. 
In recent years, several efficient algorithms
have been developed to solve the broadcasting problem in this model.

Motivated by replicated databases and peer-to-peer networks, Berenbrink et al.,
ICALP 2010, considered the so-called gossiping problem in the random phone call
model. There, each node starts with its own message and all messages have to be
disseminated to all nodes in the network. They showed that any $O(\log n)$-time
algorithm in complete graphs requires $\Omega(\log n)$ message transmissions
per node to complete gossiping, with high probability, while it is known that
in the case of broadcasting the average number of message transmissions per
node is $O(\log \log n)$. Furthermore, they explored different possibilities
on how to reduce the communication overhead of randomized gossiping in complete
graphs.

It is known that the $O(n \log \log n)$ bound on the number of message
transmissions produced by randomized broadcasting in complete graphs cannot be
achieved in sparse graphs even if they have best expansion and connectivity
properties. In this paper, we analyze whether a similar influence of the graph
density also holds w.r.t.~the performance of gossiping. We study analytically
and empirically the communication overhead generated by gossiping algorithms 
w.r.t.~the random phone call model in
random graphs and also consider simple modifications of the random phone call
model in these graphs. Our results indicate that, unlike in broadcasting, there
seems to be no significant difference between the performance of randomized gossiping in
complete graphs and sparse random graphs. Furthermore, our simulations illustrate
that by tuning the parameters of our algorithms, we can significantly reduce
the communication overhead compared to the traditional push-pull approach in
the graphs we consider.

\end{abstract}

\section{Introduction}
\subsection{Motivation}

Information dissemination is a fundamental problem in parallel and distributed
computing. Given a network, the goal is to spread one or several messages 
efficiently among
all nodes of the network. This problem has extensively been analyzed in
different communication models and on various graph classes. When talking about
information dissemination, we distinguish between \emph{one-to-all}
communication called broadcasting and \emph{all-to-all} communication called
gossiping. Much of the work devoted to information dissemination refers to the
broadcasting problem. That is, a distinguished node of the network possesses a
piece of information, which has to be distributed to all nodes in the system.
In gossiping, every node has its own piece of information, and all these
messages must be distributed to all other nodes in the network. Efficient
algorithms for gossiping are applied, e.g., in routing, maintaining consistency
in replicated databases, multicasting, and leader election, see
\cite{BT89,FL94,HKPRU05}.

There are two main approaches to design efficient algorithms for broadcasting
or gossiping. One way is to exploit the structural properties of the networks
the protocols are deployed on to design efficient
\emph{deterministic} schemes \cite{HKPRU05}. While the resulting protocols are
usually (almost) optimal, they are often not fault tolerant (note that there
are also deterministic schemes, which have polylogarithmic running time on the
graphs we consider and are highly robust, see \cite{Hae13}). Another approach
is to design simple randomized algorithms, which are inherently fault tolerant
and scalable. Prominent examples of such algorithms are based on the so-called
\emph{random phone call model}, which has been introduced by Demers et
al.~\cite{DGHILSSST87} and analyzed in detail by Karp et al.~\cite{KSSV00}. The
algorithms in this model are synchronous, i.e., the nodes act in synchronous
steps. In each step every node opens a communication channel to a randomly
chosen neighbor. The channel can then be used for bi-directional communication
to exchange messages between the corresponding nodes. It is assumed
that the nodes may decide which messages to send (they are allowed to send none
of their messages in some step), and are able to combine several messages to
one single packet, which can then be sent through a channel. Clearly, one
question is how to count the message complexity if several pieces of
information are contained in such a packet; we will come back to this question
later.

Karp et al.~motivated their work with consistency issues in replicated
databases, in which frequent updates occur. These updates must be disseminated
to all nodes in the network to keep the database consistent. They
analyzed the running time and number of message transmissions produced by
so-called push and pull algorithms w.r.t.~one single message in complete
graphs. In order to determine the communication overhead, they counted the
number of transmissions of this message through the links in the network. 
They argued that since updates occur frequently nodes have to open
communication channels in each step anyway. Thus, the cost of opening
communication channels amortizes over the total number of message
transmissions.

Motivated by the application above, Berenbrink et al.~considered the gossiping
problem \cite{BCEG10}. They assume that sending a packet through an open
channel is counted once, no matter how many messages are contained in this
packet. However, nodes may decide not to open a channel in a step, while
opening a communication channel is also counted for the communication
complexity. The first assumption is certainly unrealistic in scenarios, in which
all original messages of the nodes have to be disseminated to all other nodes; although 
network coding might overcome the inpracticability of this assumption 
in certain applications (see e.g.~\cite{Hae12}). On the
other side, in the case of leader election, aggregate computation
(e.g.~computing the minimum or the average), or consensus the assumption above
might be feasible, since then the size of the exchanged messages can
asymptotically be bounded to the size of a single message.

The algorithms developed so far in the random phone call model use so-called
push and pull transmissions. As described above, the nodes open communication
channels to some (randomly) selected neighbors. If a message is sent from the
node which called on a neighbor and initiated the communication, then we talk
about a push transmission (w.r.t.~that message). If the message is transmitted
from the called node to the one that opened the channel, then we talk about a
pull transmission. 

Although the time complexity has extensively been analyzed on various networks
in the past, the message complexity was mainly studied in complete graphs. The
question is, whether the results known for complete graphs also hold in sparse
networks with very good expansion and connectivity properties. Such networks
naturally arise in certain real world applications such as peer-to-peer
systems \cite{BW01,Gnutella}. In the case of broadcasting, it is known that the
performance of push-pull algorithms in complete graphs cannot be achieved in
random graphs of small or moderate degree \cite{Els06}. This, however, seems
not to be the case w.r.t.~gossiping. As we show in this paper, concerning the
number of message transmissions the performance of the algorithms developed in
\cite{BCEG10} can be achieved in random graphs as well. Regarding the impact of
the graph density on the {\em running time} a similar study has been done by
Fountoulakis et al.~\cite{FHP10}. They showed that there is almost no
difference between the running time of the push algorithm in complete graphs
and random graphs of various degrees, as long as the expected degree is
$\omega(\log n)$.\footnote{In this paper, $\log{n}$ denotes the logarithm of
$n$ to base $2$.}

\subsection{Related Work}

A huge amount of work has been invested to analyze information dissemination in
general graphs as well as some special network classes. We only concentrate
here on randomized protocols which are based on the random phone call model.
This model has been introduced by Demers et al.~\cite{DGHILSSST87} along with a
randomized algorithm that solves the problem of mutual consistency in
replicated databases. 

Many papers analyze the running time of randomized broadcasting algorithms that
only use push transmissions. To mention some of them, Pittel \cite{Pit87}
proved that in a complete graph a rumor can be distributed in
$\log_2(n)+\ln(n)+\BigO{1}$ steps. Feige et~al.~\cite{FPRU90} presented optimal
upper bounds for the running time of this algorithm in various graph classes
including random graphs, bounded degree graphs, and the hypercube.
 
In their paper, Karp et al.~\cite{KSSV00} presented an approach that requires
only $\BigO{\log{n}}$ time and $\BigO{n\log\log{n}}$ message transmissions,
\whp, which is also shown to be asymptotically optimal. This major improvement
is a consequence of their observation that an algorithm that uses only pull
steps is inferior to the push approach as long as less than half of the nodes
are informed. After that, the pull approach becomes significantly better. This
fact is used to devise an algorithm that uses both, push and pull operations,
along with a termination mechanism.

The random phone call model as well as some variants of it have also been
analyzed in other graph classes. We mention here the work of Chierichetti et
al.~and Giakkoupis \cite{CLP10,Gia11} who related the running time of push-pull
protocols to the conductance of a graph; or the work of Giakkoupis and
Sauerwald \cite{GS12,Gia14} on the relationship between push-pull and vertex
expansion. To overcome bottlenecks in graphs with small conductance,
Censor-Hillel and Shachnai used the concept of weak conductance in order to
improve the running time of gossiping \cite{CS12}. Earlier results related
randomized information dissemination to random walks on graphs, see
e.g.~\cite{MS06,ES09}. Modifications of the random phone call model resulted in
an improved performance of randomized broadcasting w.r.t.~its communication
complexity in random graphs \cite{ES08} and its running time in the
preferential attachment model \cite{DFF11}. The basic idea of these
modifications is used in Section \ref{sect:memory}. 

Randomized gossiping in complete graphs has been extensively studied by
Berenbrink et al.~\cite{BCEG10}. In their paper, they provided a lower bound
argument that proves $\BigOmega{n\log{n}}$ message complexity for any
$\BigO{\log{n}}$ time algorithm. This separation result marks a cut between
broadcasting and gossiping in the random phone call model. Furthermore, the
authors gave two algorithms at the two opposite points of the time and message
complexity trade-off. Finally, they slightly modified the random phone call
model to circumvent these limitations and designed a randomized gossiping
protocol which requires $\BigO{\log{n}}$ time and $\BigO{n\log\log{n}}$ message
transmissions. 

Chen and Pandurangan \cite{CP12} used gossiping algorithms for computing
aggregate functions in complete graphs (see also \cite{KDG03}). They showed a
lower bound of $\BigOmega{n\log{n}}$ on the message complexity regardless of
the running time for any gossiping algorithm. However, for this lower bound
they assumed a model that is slightly weaker than the one used in this paper.
In the main part of their paper, they presented an algorithm that performs
gossiping in $\BigO{\log{n}}$ time using $\BigO{n\log\log{n}}$ messages by
building certain communication trees. Furthermore, they also designed gossip
protocols for general graphs. For all these algorithms, they assumed a
communication model which is more powerful than the random phone call model.

Another interesting application of randomized gossiping is in the context of
resilient information exchange \cite{AGGZ10}. Alistarh et al.~proposed an
algorithm with an optimal $\BigO{n}$ communication overhead, which can tolerate
oblivious faults. For adaptive faults they provided a gossiping algorithm with
a communication complexity of $\BigO{n \log^3 n}$. Their model, however, is
stronger than the random phone call model or some simple variants of it.

Random graphs first appeared in probabilistic proofs by Erd\H{o}s and R\'{e}nyi
\cite{ER59}. Much later, they were described in the works by Bender and
Canfield \cite{BC78}, Bollob\'{a}s \cite{Bol80} and Wormald
\cite{Wor81a,Wor81b}. Aiello et~al.\ generalized the classical random graph
model introducing a method to generate and model power law graphs \cite{WCL01}.
The properties of Erd\H{o}s-R\'enyi graphs have been surveyed by
Bollob\'as \cite{Bol01}. Various properties of random graphs, including random
regular graphs, were presented in \cite{Wor99}. In recent years, random graphs
were also analyzed in connection with the construction and maintenance of large
real world networks, see e.g.~\cite{KMG03}.

\subsection{Our Results} \label{sec:ourres}

In this paper, we extend the results of \cite{BCEG10} to random graphs with
degree $\BigOmega{\log^{2+\epsilon} n}$ where $\epsilon > 0$ can be an
arbitrary constant. In \cite{BCEG10} the authors first proved a lower bound,
which implies that any address-oblivious algorithm in the random phone call
model with running time $\BigO{\log n}$ produces a communication overhead of at
least $\BigOmega{n \log n}$ in complete graphs. On the other side, it is easy
to design an $\BigO{\log n}$-time algorithm, which generates $\BigO{n \log n}$
message transmissions. The first question is whether increasing the running
time can decrease the communication overhead. This has been answered positively
for complete graphs. That is, in \cite{BCEG10} an algorithm with running time
$\BigO{\log^2 n/\log \log n}$ and message complexity $\BigO{n\log n/\log \log
n}$ was presented. However, it is still not clear whether this result can be
achieved in sparser graphs as well. One might intuitively think that results
obtained for complete graphs should be extendable to sparse random graphs as
well, as long as the number of time steps is less than the smallest degree.
However, in the related model of randomized broadcasting there is a clear
separation between results achievable in complete graphs and in random graphs
of degree $n^{\LittleO{1/\log \log n}}$ (cf.~\cite{KSSV00,Els06}). In this
paper we show that in random graphs one can obtain the same improvement on the
number of message transmissions w.r.t.~the algorithms studied so far as in
complete graphs. In light of the fact that in the slightly different
communication model analyzed by Chen and Pandurangan in their lower bound
theorem \cite{CP12} such an improvement is not even possible in complete
graphs, our result extends the evidence for a non-trivial advantage (i.e., the
possibility to improve on the communication overhead by increasing the running
time) of the well-established random phone call model to sparse random graphs.
Furthermore, we will present a modification of this model -- as in
\cite{BCEG10} -- to derive an $\BigO{\log n}$-time algorithm, which produces
only $\BigO{n \log \log n}$ message transmissions, \whp\footnote{\emph{\Whp}
means a probability of at least $1-n^{-\Omega(1)}$.}, and analyze the
robustness of this algorithm.

In this paper, we show our first result w.r.t.~the configuration model (see
next section), while the second result is proved for Erd\H{o}s-R\'enyi graphs.
Nevertheless, both results can be shown for both random graph models, and the
proof techniques are the same. Here we only present one proof w.r.t.~each graph
model. 

In our analysis, we divide the execution time of our algorithms into several
phases as in the case of complete graphs. Although the algorithms and the
overall analysis are in the same spirit as in \cite{BCEG10}, we encountered
several differences concerning the details. At many places, results obtained
almost directly in the case of complete graphs required additional
probabilistic and combinatorial techniques in random graphs. Moreover we
observed that, although the overall results are the same for the two graph
classes, there are significant differences in the performance of the
corresponding algorithms in some of the phases mentioned before. This is due to
the different structures we have to deal with in these two cases. To obtain our
results, it was necessary to incorporate these structural differences into the
dynamical behavior of the gossiping algorithms. For the details as well as a
high level description of our algorithms see Sections \ref{sect:main-result}
and \ref{sect:memory}

\section{Model and Annotation} \label{sect:model}

We investigate the gossiping problem in the random phone call model in which
$n$ players are able to exchange messages in a communication network. In our
first model, we use a Erd\H{o}s-R\'enyi random graph $G = G(n, p) = (V, E)$ to
model the network where $V$ denotes the set of players and $E \subseteq V
\times V$ is the set of edges. In this model, we have a probability of $p$ that
for two arbitrary nodes $v_1, v_2 \in V$ the edge $(v_1, v_2)$ exists,
independently.
Let $d$ denote the expected degree of an arbitrary but fixed node $v$. In this
paper, we only consider undirected random graphs for which $d \geq
\log^{2+\epsilon} n$. In this model the node degree of every node is
concentrated around the expectation, \ie{}, $d_v = \deg(v) =
d\cdot\left(1\pm\LittleO{1}\right)$, \whp. 

We also investigate the so-called configuration model introduced in
\cite{Bol80}. We adapt the definition by Wormald \cite{Wor99} as follows.
Consider a set of $d\cdot n$ edge \textit{stubs} partitioned into $n$ cells
$v_1, v_2, \dots, v_n$ of $d$ stubs each. A perfect matching of the stubs is
called a \textit{pairing}. Each pairing corresponds to a graph in which the
cells are the vertices and the pairs define the edges. A pairing can be
selected uniformly at random in different ways. \Eg, the first stub in the pair
can be chosen using any arbitrary rule as long as the second stub is chosen
uniformly at random from the remaining unpaired stubs. Note that this process
can lead to multiple edges and loops. However, \whp the number of such edges is
a constant \cite{Wor99}. 
In
our analysis we apply the principle of deferred decisions \cite{MR95}. That is,
we assume that at the beginning all the nodes have $d$ stubs which are all
unconnected. If a node chooses a link for communication for the first time in a
step, then we connect the corresponding stub of the node with a free stub in
the graph, while leaving all the other stubs as they are.

We furthermore assume that each node has an estimation of $n$, which is
accurate within constant factors. In each step, every node $v$ is allowed to
open a channel to one of its neighbors denoted by $u$ chosen uniformly at
random (in \autoref{sect:memory} we consider a simple modification of this
model). This channel is called outgoing for $v$ and incoming for $u$. We assume
that all open channels are closed at the end of every step. Since every node
opens at most one channel per step, at most one outgoing channel exists per
node.

Each node has access to a global clock, and all actions are performed in
parallel in synchronous steps. At the beginning, each node $v$ stores its
original message $m_v(0) = m_v$. Whenever $v$ receives messages, either over
outgoing channels
or over incoming channels,
these messages are combined together.
That is, $v$ computes its message in step $t$ by successively combining all
known messages together, resulting in $m_v(t) = \bigcup_{i=0}^{t-1}
m_v^{\text{(in)}}(i)$, where $m_v^{\text{(in)}}(i)$ denotes the union of all
incoming (i.e., received) messages over all connections in a step $i$ (with
$m_v^{\text{(in)}}(0) = m_v$). This combined message is used for any
transmission in step~$t$. We will omit the step parameter $t$ and use $m_v$ to
denote the node's message if the current step is clear from the context.

\section{Traditional Model} \label{sect:main-result}
\begin{algorithm}
\SetAlgoVlined
\SetKwFunction{Push}{push}
\SetKwFunction{Pull}{pull}
\SetKwFunction{PushPull}{pushpull}
\SetKwFunction{Add}{add}
\SetKwFunction{Empty}{empty}
\SetKwFunction{Pop}{pop}
\SetKwFunction{Moves}{moves}
\SetKwFor{At}{at each}{do in parallel}{end}
\SetKwFor{ForEach}{for each}{do}{end}
\SetKwFor{With}{with probability}{do}{end}
\SetKwProg{AlgorithmPI}{Phase \RNum{1}}{}{end}
\SetKwProg{AlgorithmPII}{Phase \RNum{2}}{}{end}
\SetKwProg{AlgorithmPIII}{Phase \RNum{3}}{}{end}
\SetKw{False}{false}
\SetKw{True}{true}
\SetKw{KwLet}{let}

\SetKwFor{AtNode}{at each node}{do in parallel}{end}
\SetKwFor{ForStep}{for}{do}{end}
\SetKwFor{ForRound}{for round}{do}{end}

\AlgoPhase{1}
\ForStep{$t = 1$ \KwTo $12\log{n}/\log\log{n}$}{
	\At{node $v$}{
		\Push{$m_v$}\;
	}
}
\AlgoPhase{2}
\KwLet{$\ell$ denote a large constant}\;
\ForRound{$r = 1$ \KwTo $4\log{n}/\log\log{n}$}{
	\At{node $v$}{
		\With{$ \ell / \log{n}$}{
			\Push{$m_v$}\tcp*[r]{start a random walk}
		}
	}
	\For{step $t = 1$ \KwTo $6\ell \log{n}$}{
		\At{node $v$} {
			
			\ForEach{incoming message $m'$ with \Moves{$m'$} $\leq c_{\text{moves}}\cdot\log{n}$}{
				$q_v$\texttt{.}\Add{$m' \cup m_v$}\;
				$m_v \gets m_v \cup m'$\;
			}
			\If{$\lnot$\,\Empty{$q_v$}}{
				\Push{$q_v$\texttt{.}\Pop{}}\;
			}
		}
	}
	\ForEach{node $v$}{
		\If{$\lnot$\,\Empty{$q_v$}}{
			$v $ becomes active;
		}
	}
	\For{step $t = 1$ \KwTo $1/2 \cdot \log\log{n}$} {
		\At{node $v$}{
			\If{$v$ is active}{
				\Push{$m_v$}\;
			}
			\If{$v$ has incoming messages}{
				$v $ becomes active\;
			}
		}
	}
	All nodes become inactive\;
}
\AlgoPhase{3}
\For{$t = 1$ \KwTo $8 \log{n}/\log\log{n}$}{
	\At{node $v$}{
		\PushPull{$m_v$}\;
	}
}

\caption{The \texttt{fast-gossiping} algorithm. \texttt{Push} and \texttt{pull}
operations are preceded and followed by opening and closing channels,
respectively.}
\label{alg:fast-gossiping}
\end{algorithm}

In this section we present our algorithm to solve the gossiping problem. This
algorithm is an adapted version of \texttt{fast-gossiping} presented in
\cite{BCEG10}. It works in multiple phases, starting with a distribution
process, followed by a random walk phase and finally a broadcasting phase.
These phases are described below.  Each phase consists of several rounds which
may again consist of steps. The algorithm uses the following per-node
operations.

\begin{table}[H]
\centering
\begin{tabular}{l@{ -- }l}
\texttt{push\!($m$)} & send $m$ over the outgoing channel\\
\texttt{pull\!($m$)} & \textbf{send} $m$ over \textbf{incoming} channel(s) (cf. \cite{KSSV00})\\
\texttt{pushpull} & a combination of \texttt{push} and \texttt{pull}\\
\end{tabular}
\end{table}


In Phase \RNum{2} of \autoref{alg:fast-gossiping} we require each node to store
messages associated with incoming random walks in a queue $q_v$ which we assume
to support an \texttt{add} operation for adding a message at the end and a
\texttt{pop} operation to remove the first message. The current queue status
can be obtained via the \texttt{empty} operation which yields a Boolean value
indicating whether the queue is empty or not. We furthermore assume that each
incoming message $m$ in this phase has a counter \texttt{moves}$(m)$ attached that indicates how many
real \emph{moves} it had already made. This counter can be accessed using the
\texttt{moves} operation and is described in more detail in the random walks
section.

\noindent Our main result follows.

\begin{theorem} \label{thm:theorem1}
The gossiping problem can be solved in the random phone call model on a random
graph with expected node degree $\BigOmega{\log^{2+\epsilon}n}$ in
$\BigO{\log^2n/\log\log n}$ time using $\BigO{n \log n / \log\log {n}}$
transmissions, \whp.
\end{theorem}

\Phase{1}{Distribution}
The first phase consists of $12\log{n} / \log\log{n}$ steps. In every step,
each node opens a channel, pushes its messages, and closes the communication
channel. Clearly, this phase meets the bounds for runtime and message
complexity.

Let $k \geq 6$ denote a constant. We prove our result with respect to the
configuration model described in \autoref{sect:model}. After the first phase,
we have at least $\log^k n$ informed nodes w.r.t.~each message, \whp. We
analyze our algorithm throughout this section with respect to one single message $m$ and at the end use
a union bound to show that the result holds \whp for all initial messages.

\begin{definition}
Let $I_m(t)$ be the set of vertices that are informed of message $m$ in a step
$t$, \ie{}, vertices in $I_m(t)$ have received $m$ in a step prior to $t$.
Accordingly, $|I_m(t)|$ is the number of informed nodes in step $t$. Let
$H_m(t)$ be the set of uninformed vertices, \ie{}, $H_m(t) = V \setminus
I_m(t)$.
\end{definition}

We now bound the probability that during a communication step an arbitrary but
fixed node opens a connection to a previously informed vertex, \ie{}, the
communication is redundant and thus the message is wasted. Let $v$ denote this
vertex with corresponding message $m_v$.

At the beginning, we consider each connection in the communication network as
unknown, successively pairing new edges whenever a node opens a new connection
(see principle of deferred decisions in \autoref{sect:model}). Note, however,
that this is only a tool for the analysis of our algorithm and does not alter
the underlying graph model. We observe that each node has $d_v$ communication
stubs with $\log^{2+\epsilon}{n} \leq d_v < n$. We consider a stub
\emph{wasted} if it was already chosen for communication in a previous step.
Since throughout the entire first phase each node opens at most
$12\log{n}/\log\log{n}$ channels, there still will be $\BigTheta{d_v}$
\emph{free} stubs available \whp. Observe that the number of stubs that are additionally
paired due to incoming channels can be neglected using a simple balls-into-bins argument \cite{RS98}.
If a node chooses a free stub, it is paired with
another free stub chosen uniformly at random from the graph $G$.

\begin{lemma} \label{claim:distributionA-4}
After the distribution phase, every message is contained in at least
$\log^k{n}$ nodes, \whp, where $k \geq 6$ is a constant.
\seeProof{appendix:omitted-proofs-1}
\end{lemma}

To show \refLemma{claim:distributionA-4} which corresponds to Phase I of
\autoref{alg:fast-gossiping} we first state and show
Lemmas \ref{claim:distributionA-1}, \ref{claim:distributionA-2A},
\ref{claim:distributionA-2}, and \ref{claim:distributionA-3}.

\begin{lemma} \label{claim:distributionA-1}
The probability that an arbitrary but fixed node $v$ opens a connection to a
previously uninformed vertex w.r.t.~message $m$ is at least $1 -
\BigO{\log^{-1}{n}}$. \seeProof{appendix:omitted-proofs-1}
\end{lemma}

\begin{proof}
The first phase runs for $12\log{n}/\log\log{n}$ steps with the goal to reach
at least $\log^kn$ informed vertices. We apply the principle of deferred
decision as described in \autoref{sect:model} to bound the number of
\emph{uncovered} (wasted) stubs that have already been connected. The total
number of uncovered stubs at a node can be bounded by $S=O(\log{n})$ \whp
applying a simple balls-into-bins argument \cite{RS98}. Then,
\begin{equation*}
\Probability{v \text{ chooses a \emph{wasted} stub}} \leq \frac{\BigO{\log{n}}}{ d_v} \enspace .
\end{equation*}

If in step $t$ a \emph{free} stub is chosen, the probability that the
corresponding communication partner $u$ has already been informed (or will be
informed in exactly the same step) can be bounded by
\begin{equation*}
\Probability{u \text{ is informed}} \leq \frac{d\cdot|I_m(t)|}{\left(d-S\right)n} \enspace .
\end{equation*}
Therefore, the probability $p'$ that $v$ opens a connection to an uninformed
communication partner and thus spreads the message to an uninformed node is 
\begin{equation*}
p' \geq \Probability{v \text{ chooses a \emph{free} stub to } u} \cdot \Probability{u \text{ is uninformed}}
\end{equation*}
which yields for sufficiently large $n$
\begin{equation*}
p' \geq \left(1 - \frac{1}{\log{n}}\right)\cdot \left( 1 - \frac{d\cdot\log^k{n}}{\left(d-S\right)n} \right) \geq 1 - \BigO{\frac{1}{\log n}} \enspace .\qedhere
\end{equation*}
\end{proof}

\begin{lemma} \label{claim:distributionA-2A}
Let $C$ denote a large constant. After the first $T = 4\log{n}/\log\log{n}$
steps, at least $C$ of nodes are informed of message $m_v$
\whp. 
\end{lemma}

\begin{proof}
During these first $4\log{n}/\log\log{n}$ steps we aim to reach at least $C$
informed nodes \whp. Therefore, we have a probability of at most $C/d_v$ that an informed node $v$
opens a connection to another informed node and thus causes redundant
communication. Furthermore, the probability that in an
arbitrary but fixed step $t$ every communication attempt fails and every node
$v \in I_m(t)$ performs only redundant communication can also be upper bounded
by 
$C/d$.

We define an indicator random variable $X_i$ as
\begin{equation*}
	X_{i} = \begin{cases}
		1 & \text{if } |I_m(i+1)| \geq |I_m(i)| + 1 \\
		0 & \text{otherwise}
	\end{cases}
\end{equation*}
which we sum up to obtain the number of informed nodes $X
= \sum_{i=1}^{T}X_i$. We then bound the probability that more than $C$ steps
fail, \ie{}, the number of successful transmissions $X$ is smaller than $C$, as
\begin{align*}
\Probability{X \leq C}
& \leq \sum_{i=0}^{C}\binom{T}{i} \cdot \left( 1 - \frac{1}{d} \right)^{i} \cdot \left( \frac{C}{d} \right)^{T-i}\\
& < \sum_{i=0}^{C}\left({\frac{4\log{n}\cdot e}{\log\log{n}\cdot i}}\right)^{i} \cdot \left( \frac{C}{\log^{2+\epsilon}n} \right)^{\frac{4\log{n}}{\log\log{n}}-i}\\
& \ll \frac{1}{n^2}
\end{align*}
where in the second inequality we used that $\binom{T}{i} \leq \left(\frac{T\cdot e}{i}\right)^i$.
\end{proof}

\begin{lemma} \label{claim:distributionA-2}
Let $t \in \left[4\log{n}/\log\log{n},
12\log{n}/\log\log{n}\right]$ denote an arbitrary but fixed step.
Then $|I_m(t+1)| \geq 1.5 \cdot |I_m(t)|$
with probability at least $1 - \log^{-1-\BigOmega{1}}n$. 
\seeProof{appendix:omitted-proofs-1}
\end{lemma}

\begin{proof}
According to \autoref{claim:distributionA-2A} we have $C \leq |I_m(t)| \leq \log^{k}{n}$ where $C$ denotes a large constant. In each step, every node opens a
connection to a randomly chosen communication partner. Let $c$ denote a constant. According to 
\refLemma{claim:distributionA-1}, this attempt to inform a new node fails with a
probability smaller than $c/\log{n}$. We now define the
indicator random variable $X_{i}$ for $v_i \in I_m(t)$ as follows.
\begin{equation*}
	X_{i} = \begin{cases}
		1 & \text{if $v_i$ opens a connection to $u \in I_m(t)$} \\
		0 & \text{otherwise.}
	\end{cases}
\end{equation*}
The aggregate random variable $X = \sum_{i=1}^{|I_m(t)|}X_{i}$ with expected value
$\Expected{X} \leq c\cdot |I_m(t)|/\log{n}$ represents the total number of failed
communication attempts. Clearly, we get $|I_m(t+1)| = 2|I_m(t)| - X$. Therefore, we
upper bound $X$, using Equation 12 from \cite{HR90} as follows:
\begin{align*}
\Probability{X \geq \frac{1}{2}|I_m(t)|} & \leq \left( \frac{2c}{\log{n}} \cdot 2\left(1-\frac{c}{\log{n}}\right) \right)^{|I_m(t)|/2} \\
& \leq \left( \frac{4c}{\log{n}} \right)^{|I_m(t)|/2}
\end{align*}
We can now apply the lower bound for the number of informed nodes, $|I_m(t)| \geq C$, and obtain for large $n$
\begin{equation*}
\Probability{|I_m(t+1)| \geq 1.5\cdot|I_m(t)|} \geq 1 - \log^{-C/2+1}{n} \enspace . \qedhere
\end{equation*}
\end{proof}

\begin{lemma} \label{claim:distributionA-3}
 At least $4\log{n}/\log\log{n}$ attempts out of the
$8\log{n}/\log\log{n}$ last steps in Phase~\RNum{1} succeed such that $|I_m(t+1)|
\geq 1.5 \cdot |I_m(t)|$, \whp. That is, half of the steps lead to an exponential growth. 
\end{lemma}

\begin{proof}
As of \refLemma{claim:distributionA-2}, the growth in each step can be
lower bounded by $|I_m(t+1)|/|I_m(t)| \geq 1.5$ with probability at
least $1 - \log^{-1-\BigOmega{1}}{n}$.
We now define the indicator random variable $X_i$ as
\begin{equation*}
	X_{i} = \begin{cases}
		1 & \text{if } |I_m(i+1)| < 1.5 \cdot |I_m(i)| \\
		0 & \text{otherwise.}
	\end{cases}
\end{equation*}
We sum up these indicator random variables and obtain the random variable $X
 = \sum_{i=1}^{8\log{n}/\log\log{n}}X_i$ which represents the
number of steps that fail to inform a sufficiently large set of new nodes.
Again, we use Equation 12 from \cite{HR90} to bound $X$ as follows.
\begin{align*}
\Probability{X \geq \frac{4\log{n}}{\log\log{n}}} &\leq
 \left( \frac{4}{\log^{1+\BigOmega{1}}{n}} \left( 1 - \frac{1}{\log^{1+\BigOmega{1}}n}\right) \right)^{\frac{4\log{n}}{\log\log{n}}}\\
& \ll n^{-3} \qedhere
\end{align*}
\end{proof}

We now combine these results to give a proof for
\refLemma{claim:distributionA-4} which concludes the first phase.

\begin{proof}[Proof of \refLemma{claim:distributionA-4}]
Since each message $m$ starts in its original node, $|I_m(0)| =
1$. We conclude from \refLemma{claim:distributionA-3} that \whp
in at least $4\log{n}/\log\log{n}$ steps the number of nodes informed of $m$ increases by a factor of at least $1.5$ as long as $|I_m(t)| \leq \log^k n$.
Thus, we have \whp
\begin{equation*}
\left|I_m\left(\frac{12\log{n}}{\log\log{n}}\right)\right| \geq \min\left\{\log^kn, 1.5^{\frac{4\log{n}}{\log\log{n}}}\right\}
= \log^{k}{n} \enspace .
\end{equation*}
We apply a union bound over all messages and the lemma follows.
\end{proof}

\Phase{2}{Random Walks}
After the first phase, each message is contained \whp in at least $\log^k{n}$
nodes, where $k \geq 6$ is a constant. We aim to reach
$n\cdot2^{-\log{n}/\log\log{n}}$ informed nodes for each message in the
second phase and therefore assume for any message $m$ and any step $t$ in Phase~\RNum{2} 
that $\log^k{n} \leq |I_m(t)| \leq n\cdot2^{-\log{n}/\log\log{n}}$.

At the beginning of Phase~\RNum{2} a number of nodes start so-called random
walks. If a random walk arrives at a node in a certain step then this node adds
its messages to the messages contained in the random walk and performs a push
operation, \ie{}, the random walk moves to a neighbor chosen uniformly at
random. This is done for $\BigO{\log{n}}$ steps.
To ensure that no random walk is lost, each node collects all incoming
messages (which correspond to random walks) and stores them in a queue to send
them out one by one in the following steps. The aim is to first collect and
then distribute messages corresponding to these walks. After the random walk
steps all nodes containing a random walk become \emph{active}. A
broadcasting procedure of $1/2\cdot\log\log{n}$ steps is used to increase the
number of informed nodes by a factor of $\BigTheta{\sqrt{\log{n}}}$. The entire
second phase runs in $4\log n/\log \log n$ rounds which correspond to the
outer \texttt{for}-loop in Phase~\RNum{2} of \autoref{alg:fast-gossiping}. Each round consists
of $\BigO{\log n}$ steps. Thus, the running time of this phase is in
$\BigO{\log^2n/\log\log n}$.

Note that although random walks carry some messages, we assume in our analysis
that the nodes visited by the random walks do not receive these messages from
the random walks. That is, the nodes are not necessarily informed
\emph{after} they were visited by a random walk and thus are not accounted to
$I_m$.

In the following, we consider an arbitrary but fixed round $r$ that belongs to the second phase with $1 \leq r
\leq 4\log{n}/\log\log{n}$. Whenever we use the expression $I_m(r)$, we mean the set of informed
nodes at the beginning of the corresponding round $r$, even though the informed
set may be larger in some step of this round.

At the beginning of each round, every node flips a coin. With a probability of
$\ell / \log{n}$, where $\ell$ denotes a large constant, the node
starts a random walk. We first need to bound the total number of
random walks which are initiated. As their number does not depend on the
underlying graph, we can use the result of \cite{BCEG10} for the number of
random walks 
and obtain
$\BigTheta{n/\log{n}}$ random walks \whp. Therefore, the bounds on the message
complexity of $\BigO{n\log{n}/\log\log{n}}$ are met during the random walks phase.
In the following we only consider random walks that carry an arbitrary but fixed
message $m$.

We observe that these random walks are not independent from each other, since a
random walk $w$ incoming at node $v$ is enqueued into a queue $q_v$. Therefore,
$w$ may be delayed before it is sent out again by $v$ and this delay is based
on the number of (other) random walks that are currently incident at node $v$.
If $v$ eventually sends out the random walk $w$, we say $w$ makes a
\emph{move}. It is now an important observation that the actions of the random walks in a specific step
are \textbf{not} independent from each other. Their \emph{moves},
however, are.

Now a question that arises naturally is whether the number of moves made by an
arbitrary but fixed random walk $w$ is large enough to \emph{mix}. This
question is covered in \autoref{lem:random-walks-mix}, where we will argue that
the number of moves taken by every random walk is $\BigOmega{\log{n}}$ and
therefore larger than the mixing time of the network. In the following lemmas,
especially in \autoref{claim:number-of-random-walks-in-I}, we will also require
that the random walks are not correlated, which clearly is not true if we
consider the steps made by the algorithm. However, the moves of the random
walks are independent from each other. That is, after mixing time moves, the node that hosts random walk
$w$ after its $i$-th move is independent from the nodes that host any other of
the random walks after their $i$-th moves. We furthermore require, \eg{}, in
\autoref{claim:random-walks-end-in-uninformed-area}, that after some mixing
steps the random walks are distributed (almost) uniformly at random over the
entire graph. This is enforced as we stop every random walk once it has reached
$c_{\text{moves}}\cdot\log{n}$ moves for some constant $c_{\text{moves}}$. Note, that we implicitly attach a counter to each
random walk which is transmitted alongside the actual message. In the first
inner \texttt{for}-loop in Phase~II of \autoref{alg:fast-gossiping} we then
refuse to enqueue random walks that have already made enough moves.

Note that starting with \autoref{claim:number-of-random-walks-in-I}, when we talk about random walks in a certain step $i$
we always mean each random walk after its $i$-th
move. This does not necessarily have to be one single step of the algorithm,
and the corresponding random walks are scattered over multiple steps.
Since, however, the moves of the random walks are independent from each other,
the actual step can be reinterpreted in favor of the random walk's movements.
What remains to be shown is that every random walk makes indeed
$\BigOmega{\log{n}}$ moves. This is argued in the following lemma.

\begin{lemma} \label{lem:random-walks-mix}
The random walks started in Phase~II of \autoref{alg:fast-gossiping} make
$\BigOmega{\log{n}}$ moves, \whp.
\end{lemma}

\begin{proof}

\newcommand\inlinefrac[2]{#1/#2}
\let\nicefrac\frac
At the beginning we fix one single random walk $r$, and let $P$ be the sequence
of the first $\log n/4$ nodes visited by this random walk, whenever $r$ is
allowed to make a move. Note that some nodes in $P$ may be the same (e.g., the
random walk moves to a neighbor $v$ of some node $u$, and when the random walk
is allowed to make a move again, then it jumps back to $u$). Clearly, the
number of random walks moving around in the network is $\BigO{\inlinefrac{n}{\log n}}$, \whp.
For simplicity, let us assume that there
are exactly $\inlinefrac{n}{\log n}$ random walks (a few words on the general case are given at the end of this proof). We now consider the number of
vertices in the neighborhood $N(v)$ of a vertex $v$ which host a random walk
at some time step $i$. We show by induction that for each time step
$1 \leq i \leq \inlinefrac{\log n}{4}$ and any node $v$ with probability $1-\inlinefrac{2i}{n^3}$ it holds that

\begin{enumerate}

\item The number of vertices hosting at least one random walk is at most
$\nicefrac{d}{\log n}\left(1+\nicefrac{2i}{\log n}\right)$. \\ This set is
denoted by $N_1(v)$.

\item The number of vertices hosting at least two random walks is at most
$\nicefrac{d}{\log^2 n}\left(1+\nicefrac{2i}{\log n}\right)$. \\ This set is
denoted by $N_2(v)$.

\item The number of vertices hosting $3$ or more random walks is at most
$\nicefrac{di}{\log^3 n}$. \\ This set is denoted by $N_3(v)$.

\end{enumerate}

For the proof we condition on the event that there are at most $4$ circles
involving $v$, the nodes of the first neighborhood of $v$, and the nodes of the
second neighborhood of $v$. Note that this event holds with very high
probability for a large range of $d$ (i.e., $d \leq n^{\alpha}$ for some
$\alpha$ constant but small, see e.g., \cite{DFS09}, \cite{BES14}, or for
random regular graphs a similar proof done by Cooper, Frieze, and Radzik
\cite{CFR09}. These edges can be treated separately at the end and are neglected
for now. Moreover, in the configuration model it is possible to have multiple
edges or loops. However, for this range of degrees there can be at most
constant many, which are treated as the circle edges mentioned above at the end.
We use $c_{\text{circle-edges}}$ to denote the constant for the number of multiple edges and circle edges.
For the
case $d > n^{\alpha}$ different techniques have to be applied, however, a
similar proof as in the complete graph case can be conducted. For now, we assume
that $d \geq \log^5 n$. For random regular graphs of degree $d \in
[\log^{2+\epsilon}n, \log^5 n]$ the proof ideas are essentially the same, 
however, at several places an elaborate case analysis becomes necessary.

Now to the induction. In the first time step, the hypothesis obviously holds,
\ie{}, each node starts a random walk with probability $1/\log n$,
independently. Assume now that the induction hypothesis holds for some time
step $i$, and we are going to show that it also holds for step $i+1$. Note that
the assumption holds in the neighborhood of each node, and thus, also in the
neighborhoods of the nodes of $N(v)$. We start by showing claim 3. In each
step, every node of $N_1(v)$ will release a random walk. There are $d$ vertices
in $N(v)$, and $di/\log^3 n$ nodes with at most $3$ random walks.  A node of
$N(v) \setminus N_2(v)$ becomes an $N_3(v)$ node with probability at most
\begin{equation}
\binom{ \frac{d}{\log n}\left(1+\frac{2i}{\log n}\right) }{ 3 } \cdot \frac{1}{d^3} \enspace . \label{eq3}
\end{equation}
(Note that a more exact calculation involves the sum $\sum_{i=2}^{|N_1(w)|} \binom{|N_1(w)|}{3}\left(\frac{1}{d}\right)^3\left(1-\frac{1}{d}\right)^{\left(|N_1(w)|-i\right)}$ where $w$ is a neighbor of $v$. This sum can be approximated efficiently by using bounds on the tail of the binomial distribution. We work here with the simpler expression in Eq.~(\ref{eq3}).)
Therefore the expected value of these nodes is at most 
\[
\Expected{Z} = \binom{ \frac{d}{\log n}\left(1+\frac{2i}{\log n}\right) }{ 3 } \cdot \frac{1}{d^3} \cdot d \enspace .
\]
Since we only consider the nodes which are not involved in any cycles and do
not have multiple edges each node $w'$ in the second neighborhood of $v$ sends
a random walk to the corresponding neighbor in $N(v)$ independently of the
other nodes in the second neighborhood. Thus we can apply Chernoff bounds and obtain that the number
of the nodes in $N(v)$ which receive a random walk is $\Expected{Z}\left(1+o(1)\right)$.

An $N_2(v)$ node becomes an $N_3(v)$ node with probability
\[ \binom{ \frac{d}{\log n}\left(1+\frac{2i}{\log n}\right) }{ 2 } \cdot \frac{1}{d^2} \enspace .\]
Again, since the neighborhoods of the different nodes are disjoint (up to at
most $4$ edges, which can be treated separately and therefore are neglected in
the future), we may apply Chernoff bounds, and obtain an upper bound for $N_3(v)$ as follows.
\begin{align*}
& \binom{ \frac{d}{\log n}\left(1+\frac{2i}{\log n}\right) }{ 3 } \cdot \frac{1}{d^3} \cdot d \cdot \left(1+o(1)\right) + \\ 
& \binom{ \frac{d}{\log n}\left(1+\frac{2i}{\log n}\right) }{ 2 } \cdot \frac{1}{d^2} \cdot \frac{d}{\log^2n} \cdot \left(1+\frac{2i}{\log n}\right) \cdot \left(1+o(1)\right) + \\
& |N_3(v)| + 
 c_{\text{circle-edges}}
\end{align*}
Recall that we initially neglected circle edges and multiple edges. In the
worst case, the nodes incident at these edges send a random walk to $N_3(v)$ (as well as to $N_2(v)$ and $N_1(v)$) in every step
and therefore the last expression $c_{\text{circle-edges}}$ denotes a constant
for these additional incoming messages.
Noting that furthermore $i < \inlinefrac{\log n}{4}$ we obtain that $N_3(v)
\leq \inlinefrac{d(i+1)}{\log^3 n}$ in the next step, \whp.

Concerning the $N_2(v)$ nodes, a node being in $N(v) \setminus N_2(v)$ becomes an $N_2(v)$ node with probability
\[ \binom{ \frac{d}{\log n} \left(1+\frac{2i}{\log n}\right) }{ 2 } \cdot \frac{1}{d^2} \enspace .\]
Similarly, an $N_2(v)$ node will still remain in $N_2(v)$ with
probability $\nicefrac{1}{\log n}\left(1+\nicefrac{2i}{\log n}\right)$. Applying again Chernoff bounds for both cases separately, we obtain
the result. Additionally, we add the $N_3(v)$ nodes to the $N_2(v)$ nodes, and a similar calculation as above shows that the
given bound is not exceeded.

Now we concentrate on nodes in $N_1(v)$. A node being in $N(v) \setminus N_2(v)$ becomes (or remains) an $N_1(v)$ node with probability 
\[\frac{1}{\log n} \left(1+\frac{2i}{\log n}\right) \enspace .\]
Note that there can be at most $d$ nodes in this set. Applying Chernoff bounds as above and adding
the $N_2(v)$ nodes to this set, we obtain the upper bound.

Now, we know that every neighborhood $N(v)$ has at most $\nicefrac{d}{\log n}\left(1+\nicefrac{2i}{\log n}\right)$
nodes which possess at least one random walk in step $i$, \whp. This implies that in each time
step, the number of random walks sent to $v$ is a random variable which has a
binomial distribution with mean \[ \frac{1}{\log n}\left(1+\nicefrac{2i}{\log n}\right) \leq \nicefrac{3}{2\log{n}} \enspace . \]
That is, if we denote by $X_v$ this random variable then $X_v$ can be modeled by the sum of $\inlinefrac{3d}{2\log{n}}$ Bernoulli random variables with success probability $\inlinefrac{1}{d}$.
Thus, within $\inlinefrac{\log n}{4}$
steps, $v$ collects in total at most $X$ random walks, where
\begin{equation*}
 \Probability{X > \frac{3}{4} \cdot c \cdot \log n / \log\log n} 
 \leq \left( \frac{e^{\frac{c\log n}{\log\log n}-1}}{\left(\frac{c \log n}{\log\log n}\right)^{ \frac{c \log n}{\log\log n}}} \right)^{\frac{3}{4}} 
 \leq \frac{1}{n^5} \enspace ,
\end{equation*}
if the constant $c$ is large enough. This also implies that
at any node, there will be no more than $\inlinefrac{c\log n}{\log\log n}$ many random walks for some proper $c$,
and hence, if a random walk arrives, it is enough to consider the last $\inlinefrac{c\log n}{\log\log n}$ steps. That is, when a random walk arrives to a node, the number
of random walks can be represented by the sum of $\inlinefrac{c \log n}{\log\log n}$ independent
random variables $X_v$ (as described above) with binomial distribution having mean $\BigO{\inlinefrac{1}{\log n}}$ each. The
probability above is an upper bound, and this bound is independent of the distribution
of the random walks among the vertices (conditioned on the event that the
induction hypotheses 1, 2, and 3 hold, which is true with very high
probability).

Consider now some time steps $t_1, t_2, \dots, t_{i}, \dots$ which denote
movements of the random walk $r$ from one vertex to another one. Whenever $r$ makes
a move at some time $t_i$, it has to wait for at most \begin{equation}
\sum_{j=t_i - \nicefrac{c\log n}{\log \log n} + 1}^{t_i} X_j
\label{eq1}
\end{equation}
steps, where $X_j$ is a random variable having binomial
distribution with mean $O\left(\inlinefrac{1}{\log n}\right)$. One additional step after this waiting time $r$ will make a move.
If a random walk leaves a node twice, at time $t_i$ and $t_j$ respectively (with $t_i < t_j$), then we consider the sum in Eq. (\ref{eq1}) from $\max\left\{t_j-\inlinefrac{c\log{n}}{\log\log{n}}+1, t_i + 1\right\}$.
Observe that $t_i$ is a random variable that depends on $t_{i-1}$ and the random variable for above waiting time.
In order to have
\begin{equation}
\sum_{i=1}^t \left( \sum_{j=t_i-\frac{c\log n}{\log \log n}+1}^{t_i} X_j(t_i) + 1  \right) = \frac{\log n}{4}
\label{eq2}
\end{equation}
with some probability at least
$\inlinefrac{1}{n^2}$, $t$ must be $\BigOmega{\log n}$ where $X_j(t)$ is a random variable with the same properties as $X_j$ described above. This implies that
within $\inlinefrac{\log n}{4}$ steps, $r$ makes $\BigOmega{\log n}$ moves, \whp.
This holds, since \begin{equation*}
\Probability{\sum_{i=1}^t  \sum_{j=t_i-\frac{c\log n}{\log \log n}+1}^{t_i} \sum_{k=1}^{\frac{3d}{2}} X_{ijk} \geq \frac{\log{n}}{4}} = \frac{1}{n^{\omega(1)}}
\end{equation*}
for $t = \BigO{\log n}$. The sum $\sum_{k=1}^{\frac{3d}{2}} X_{ijk}$ represents the random variable $X_j(t_i)$ (see above, where $X_{ijk}$ is a Bernoulli random variable with success probability $\inlinefrac{1}{d}$) and the second sum represents the inner sum from Eq. (\ref{eq2}).
Observe that above sum represents an upper bound on the sum of the random walks that random walk $r$ meets when moving from one node to another according to the sequence $P$ defined at the beginning. That is, the sum gives the time $r$ has to wait at the nodes without making a move.

Note that in the proof we showed that if at the beginning there are
$\inlinefrac{n}{\log n}$ randomly chosen nodes starting a random walk, then each
random walk makes at least $\BigOmega{\log{n}}$ moves \whp. If we start
$\ell\inlinefrac{n}{\log{n}}$ random walks, then the proof can be adapted
accordingly so that $\BigOmega{\log{n}}$ moves are also performed by each random
walk, \whp. (The calculations become a bit more complex, however.) Noting that the eigenvalues
of the transition matrix of these graphs are inverse polynomial in $d$, the random walks
are well mixed. \qedhere 

\end{proof}

\begin{lemma} \label{claim:number-of-random-walks-in-I}
During the $\BigTheta{\log{n}}$ steps that follow the coin flip, $I_m(r)$ is visited by random walks at least
$\BigOmega{|I_m(r)|}$ times, \whp.
\end{lemma}
\begin{proof}
Let $m$ denote an arbitrary but fixed message and $I_m(r)$ the corresponding
set of vertices that are informed of $m$ at the beginning of a round $r$.
Depending on the coin flip each node starts a random walk with probability
$\ell/\log{n}$ and therefore we have a total number of random walks in
$\BigTheta{n/\log{n}}$.  Let $X$ denote the random variable for the number of
random walks that currently reside in $I_m(r)$ in an arbitrary but fixed step
of round $r$. In expectation we have $\Expected{X} = |I_m(r)| \cdot \ell /
\log n$ such random walks. We use Chernoff bounds on $X$ and obtain that
\begin{equation*}
\Probability{|X-\Expected{X}| > \frac{\Expected{X}}{\log{n}} } \leq n^{-\BigOmega{\frac{|I_m(r)|}{\log^3n}}} \enspace .
\end{equation*}
Therefore, we conclude that this number of random walks is concentrated around
the expected value \whp and thus is in $\BigTheta{|I_m(r)|/\log n}$. Since these
random walk moves are not correlated and choose their next hop uniformly at random
we conclude that in any such step the number of random walks that reside in
$I_m(r)$ is in $\BigTheta{|I_m(r)|/\log n}$ \whp. Using union bounds over all
$\BigTheta{\log n}$ steps following the coin flip we conclude that there are
$\BigTheta{|I_m(r)|}$ random walk visits in the set of informed vertices
$I_m(r)$ in these $\BigTheta{\log n}$ steps, \whp.
\end{proof}

Note that a rigorous analysis of the behavior of similar parallel random walks
on regular graphs has been already considered by Becchetti
et~al.~\cite{BCNPS15}.

These $\BigTheta{|I_m(r)|}$ random walks do not necessarily need to be distinct.
It may happen that a single random walk visits the set $I_m(r)$
multiple times, in the worst case up to $\Theta(\log n)$ times. We therefore have to give bounds
the number of random walks that visit $I_m(r)$ only a constant number of
times.

We now distinguish two cases. Let $\kappa$ denote a constant.
In the following, we consider only \textit{sparse} random graphs with expected node
degree $d$ for which
$\log^{2+\epsilon}n \leq d \leq \log^\kappa n$.
We observe that if $d \leq
\log^\kappa n$ the informed set consists of several
connected components, which we call regions, that have a diameter of at most $\BigO{\log\log{n}}$ each
and a distance between each other of at least $\BigOmega{\log\log{n}}$ (see
\refLemma{claim:random-walks-end-in-uninformed-area}).

Let $v$ denote an arbitrary but fixed vertex and let $T(v)$ denote the subgraph
induced by nodes that can be reached from $v$ using paths of length at most
$\BigO{\log\log n}$. It has been shown in Lemma 4.7 from \cite{BES14} that
$T(v)$ is a \textit{pseudo-tree} \whp, \ie{}, a tree with at most a constant
number of additional edges. Therefore, we can assign an orientation to all
edges of $T(v)$ in a natural way, pointing from the root node $v$ towards the
leafs. Thus, any edge in $T(v)$ is directed from $v_1$ to $v_2$ if $v_1$ is in
at most the same level as $v_2$. We consider edges violating the tree property
with both nodes on the same level as oriented in both ways. Whenever a random
walk takes a step that is directed towards the root of the tree, we speak of a
\emph{backward move}.

\begin{lemma} \label{claim:new-lemma-rw}
Assume $d \leq \log^\kappa n$. An arbitrary but fixed random walk leaves
the set of informed vertices $I_m(r)$ to a distance in $\BigOmega{\log\log{n}}$ and
does not return to $I_m(r)$ with a probability of at least $1 -
\log^{-2}{n}$.
\end{lemma}

To show \refLemma{claim:new-lemma-rw} we first introduce and show Lemmas
\ref{claim:walk-back-constant-time} and
\ref{claim:no-return-after-loglogn-steps}.

\begin{lemma} \label{claim:walk-back-constant-time}
Assume $d \leq \log^\kappa n$. Any random walk originating in a node of $T(v)$ takes in the first $2\log\log{n}$ steps only a constant number of
backward moves with probability at least $1-\log^{-3}n$. \seeProof{appendix:omitted-proofs-1}
\end{lemma}

\begin{proof}
We consider an arbitrary but fixed random walk $w$ that is informed with $m_v$, \ie{}, it carries $m_v$,
and focus on the first $\log\log n$ steps after $w$ was informed for the
first time. Let $X_i$ denote the random variable for the orientation of the
edge taken by $w$ in the $i$-th step, defined as
\begin{equation*}
X_i = \begin{cases}
		1 & \text{if $w$ takes a back edge in step $i$} \\
		0 & \text{otherwise.}
	\end{cases}
\end{equation*}
From the pseudo-tree-property of $T(v)$ we can conclude that the probability
of $w$ using a back edge is at most $\BigO{1/d}$, since every node has one edge to its
parent and additionally at most a constant number of edges that are
directed backwards.

Let $c\geq 3$ denote a constant. We define the random variable $X = \sum_{i=1}^{\log\log n} X_i$ for the
number of back edges taken by $w$ in $2\log\log{n}$ steps with 
expected value $\Expected{X} \leq \BigO{\log\log{n}/d}$. 
Since we can assume that $X$ has a binomial distribution we can directly derive the probability that more than a constant number of $c$
steps taken by $w$ are backward steps using $ \binom{n}{k} \leq \left(\frac{n\cdot e}{k}\right)^{k} $
as follows.
\begin{align*}
& \Probability{X \geq c} \\
&= \sum_{i=c}^{2 \log\log n} \binom{2 \log\log n}{i} \cdot \left(\frac{\BigO{1}}{d}\right)^{i} \cdot \left(1 - \LittleO{1}\right)^{2\log\log n - i} \\
&< \sum_{i=c}^{2 \log\log n} \left(\frac{2 \log\log n \cdot e}{i}\right)^{i} \cdot \left(\frac{\BigO{1}}{d}\right)^{i} \\
&< \sum_{i=c}^{2 \log\log n} \left(\frac{\BigO{\log\log n} }{i \cdot \sqrt{d}}\right)^{i} \cdot \left(\frac{1}{\sqrt{d}}\right)^{i}\\
&\leq \BigO{\log\log n} \cdot \left(\frac{\BigO{\log\log n} }{\sqrt{d}}\right)^{c} \cdot \left(\frac{1}{\sqrt{d}}\right)^{c}\\
&< \log^{-c}n \leq \log^{-3}n \qedhere
\end{align*}
%
\end{proof}

In the following we consider random walks that are more than $\log\log n$
steps away from the set of informed nodes.

\begin{lemma} \label{claim:no-return-after-loglogn-steps}
Assume $d \leq \log^\kappa n$. The probability that a random walk does not
return to an informed region $T(v)$ in $\BigO{\log n}$ steps once it has a
distance to the region that is greater than $\log\log n$ steps is at least $1 -
\log^{-3}n$. \seeProof{appendix:omitted-proofs-1}
\end{lemma}
\begin{proof}
Let $w$ denote an arbitrary but fixed random walk and let $a$ denote a constant. We use a Markov chain to model and analyze the behavior of the random walk $w$ with
respect to its distance to $I_m$. Let $X$ denote a random variable for the
number of steps $w$ takes backward. Because of the pseudo-tree property the
probability that the random walk moves backward can be bounded by $p' = \BigO{1/d}$ for any node with
distance $\BigO{\log\log n}$ to the root. Thus, the probability that $w$
takes $\tau$ backward steps in a total of $a \log n$ tries can be
bounded by 
\begin{align*}
\Probability{X = \tau} & \leq \binom{a \log{n}}{\tau} \left(\frac{4}{\log^2n}\right)^{\tau} \left(1 - \frac{4}{\log^2n}\right)^{a \log{n} - \tau} \\
& < \binom{a \log{n}}{\tau} \left(\frac{4}{\log^2n}\right)^{\tau} \\
\intertext{which gives using $\binom{n}{k} \leq \left(\frac{n\cdot e}{k}\right)^k$}
\Probability{X = \tau} & < \left( \frac{e a \log n}{\tau} \right)^\tau \left(\frac{4}{\log^2n}\right)^{\tau}  = \left( \frac{4 e a}{\tau \log{n}} \right)^\tau \enspace .\\
\intertext{We now consider only those random walks that have a distance larger than
$\log\log n$ to the root node of the local informed tree. Note that there
remains a \textit{safety belt} around the informed set, since the broadcasting
procedure performed by each random walk at the end of the round (see last \texttt{For}-loop in Phase II of \autoref{alg:fast-gossiping}) builds up a tree with height at most
$1/2\cdot\log\log n$. We investigate  $\tau = 1/2\cdot\log\log n$, the distance to cross this \textit{safety belt}, 
 and observe}
\Probability{X = \tau} & \leq \left( \frac{1}{\log{n}} \right)^{\log\log{n}/2} \\
\intertext{and therefore}
\Probability{X \geq \tau} &< \sum_{\tau = \log\log n/2}^{a \log n} \left(\frac{4 e a}{\tau \log n} \right)^\tau \\
& < \left( a \log{n} - \log\log n/2 \right) \left( \frac{1}{\log{n}} \right)^{\log\log{n}/2} \\
& \leq \log^{-3}n \enspace .\qedhere
\end{align*}
\end{proof}

We are now ready to prove \refLemma{claim:new-lemma-rw}.

\begin{proof}[Proof of \refLemma{claim:new-lemma-rw}]
From \refLemma{claim:walk-back-constant-time} and
\refLemma{claim:no-return-after-loglogn-steps} we conclude that 
with probability at least 
\begin{equation*}
1 - \left(1 - \frac{1}{\log^{3}n} \right)\left(1 - \frac{1}{\log^{3}n} \right)  \geq 1 - \log^{-2}n
\end{equation*}
an arbitrary but fixed random walk $w$ leaves the set of informed vertices to
some distance in $\BigO{\log\log n}$ and does not return. Together, this yields \refLemma{claim:new-lemma-rw}.
\end{proof}

We will show in \refLemma{claim:random-walks-end-in-uninformed-area} that the
distance between the informed regions is at least $\BigOmega{\log\log n}$.
Thus we can show \refLemma{claim:theta-f-n-distinct-random-walks} using the following definition.

\begin{definition}[Safe Area]
\label{def:safe-area}
A \emph{safe area} is a set of nodes that are uninformed and
have distance at least $\log\log n$ to any informed node.
\end{definition}

\begin{lemma} \label{claim:theta-f-n-distinct-random-walks}
Assume $d \leq \log^\kappa n$. The number of random walks that visit
$I_m(r)$ at most a constant number of times is 
$\BigTheta{|I_m(r)|}$ \whp.
\end{lemma}

\begin{proof}
Let $c$ denote a constant. We examine steps $s \in [\log n, 2\log n]$ after the coin flip. In \refLemma{claim:number-of-random-walks-in-I} we showed that the number of
random walks visits in the informed set $I_m(r)$ is in $\BigTheta{|I_m(r)|}$ \whp. Let $W$
denote this number.
Let furthermore $Q$ be the set of random walks that visit $I_m(r)$ at most a
constant number of $c$ times and let $P$ be the set containing all the other
random walks. The inequality
\begin{equation*}
W \leq c\cdot |Q| + \log{n}\cdot|P|
\end{equation*}
holds since the random walks in $Q$ hit $I_m(r)$ at most $c$ times, and the
random walks in $P$ at most $\log{n}$ times, respectively.
%
The probability
that a random walk does not leave the set of informed vertices to a distance of
$\log\log n$ can be bounded by $\log^{-3} n$ according to
\refLemma{claim:walk-back-constant-time}. Furthermore, we
need to show that with probability $\log^{-2} n $ the random walk hits
any other informed region at most a constant number of times. This follows from
the idea of a \textit{safety belt} as described in the proof of
\refLemma{claim:no-return-after-loglogn-steps}, where we observed that the
probability that a random walk returns through this region of distance
$1/2\cdot\log\log n$ to any informed node can be bounded by $\log^{-3} n$.
A simple union bound over all $\BigTheta{\log n}$
steps gives a probability of $\log^{-2} n$ that a random walk hits an informed node.

It is crucial that in above analysis we regard only informed nodes that arose
from random walks broadcasting in a \textit{safe area} according to Definition
\ref{def:safe-area},
thus giving us above setup of informed balls, safety belts and long distances
between informed regions. We show these properties in \refLemma{claim:random-walks-end-in-uninformed-area}.
Therefore, we can bound the probability that an individual
random walk visits $I_m(r)$ more often than a constant number of $c$ times by $\log^{-2}n$

We now bound the number of random walks in $P$, \ie{}, the number of
random walks that return more often than a certain constant number of $c$
times. Let the indicator random variable $X_i$ be defined for a random walk $w_i$ as
\begin{equation*} 
X_i = \begin{cases}
		1 & \text{if $w_i$ returns more often than $c$ times} \\
		0 & \text{otherwise.}
	\end{cases}
\end{equation*}
The random variable $X =|P|= \sum_{i=1}^{W} X_i$ describes the number of random
walks that return more often than $c$ times. The expected value of $X$
can be bounded by $\Expected{X} \leq W\cdot\log^{-2}n$.
Since all random walks are independent we apply Chernoff bounds and obtain for
sufficiently large $n$
\begin{align*}
\Probability{X \geq \left( 1 + \frac{1}{\log n} \right)\Expected{X}} &\leq e^{-\frac{\BigTheta{|I_m(r)|}}{3\log^{4}n} } 
\leq n^{-\LittleOmega{1}} \enspace .
\end{align*}
Therefore, \whp
\begin{align*}
W  &\leq c \cdot |Q| + \log {n} \cdot |P|\\
& \leq c \cdot |Q| + \log n \cdot \left(1 + \frac{1}{\log{n}}\right) \cdot \frac{W}{\log^{2}n}
\intertext{and thus}
c\cdot|Q| &\geq W \cdot\left(1 - \left(\frac{1}{\log{n}} + \frac{1}{\log^{2}{n}}\right)\right)
\end{align*}
which gives $|Q| = \BigTheta{W}$. Since $W = \BigTheta{|I_m(r)|}$ we
finally obtain that $|Q| = \BigTheta{|I_m(r)|}$.
\end{proof}

\begin{lemma} \label{claim:random-walks-end-in-uninformed-area}
Assume $d \leq \log^{\kappa} n$. 
The number of
random walks that terminate in a \emph{safe area} 
is in $\BigTheta{|I_m(r)|}$.
\seeProof{appendix:omitted-proofs-1}
\end{lemma}

\begin{proof}
Let $W$ denote the number of random walks. Each random walk performs at the end 
$O(\log n)$ mixing steps. Thus, the random walks are distributed
(almost) uniformly at random over the entire graph. For the analysis, we now proceed as
follows. We uncover one random walk after another, thereby omitting random
walks that stopped too close to another previously uncovered random walk. For
each of these steps, the probability $p_{\text{unsafe}}$ that a random walk ends up in an
unsafe area can be bounded as follows.
\begin{equation*}
p_{\text{unsafe}} \leq \frac{|I_m(r)| d^{ \log \log n}}{n} \leq \frac{\log^{\kappa \log \log n}n}{2^{\log n/\log \log n}}
\end{equation*}
We define for every random walk $w_i$ an indicator random variable $X_i$ as
\begin{equation*} 
X_i = \begin{cases}
		1 & \text{if $w_i$ ends in an unsafe area} \\
		0 & \text{otherwise}
	\end{cases}
\end{equation*}
and bound the random variable $X = \sum_{i=1}^WX_i$ representing the number of
random walks that end within an unsafe area. The expected number of these walks is $\Expected{X} = p_{\text{unsafe}} W$.
Since all random walks are independent,
applying Chernoff bounds yields for large $n$
\begin{equation*}
\Probability{X \geq \left(1+\frac{1}{\log n}\right)\Expected{X}} \leq e^{-\frac{\Expected{X}}{3\log^2 n}} \leq n^{-\LittleOmega{1}}\enspace .
\end{equation*}
Therefore, there are $\BigTheta{|I_m(r)|}$ random walks in safe areas \whp.
\end{proof}

\begin{lemma} \label{claim:random-walks-do-not-hit-each-other}
Assume $d > \log^\kappa n$. A random walk visits the set $I_m(r)$ at
most a constant number of times with probability at least $1 -
\log^{-2} n$. Furthermore, the  number of random walks that visit the
set $I_m(r)$ a constant number of times is
$\BigTheta{|I_m(r)|}$ \whp.
\end{lemma}

\begin{proof}
Since our algorithm runs for at most $\BigO{\log^2 n / \log\log n}$ time, each
node has during the second phase at least $\BigOmega{d}$ free stubs available.
Therefore we bound the probability that in step $t$ an arbitrary but fixed
random walk $w$ located at node $v$ opens an already used stub or connects to 
a node $u$ in the informed set $I_m(t)$ as follows.
\begin{align*}
\Probability{v \text{ opens a used stub}} & \leq \BigO{\log^{-\kappa + 2}}\\
\Probability{u \in I_m(t)} & \leq |I_m(t)| \cdot d / \left(n \left( d - \log^2{n} \right)\right)
\end{align*}
Therefore, we obtain a probability $p'$ that an unused stub
is chosen and the corresponding communication partner was not previously
informed of $p' > 1 - \log^{-3} n$. 
We apply union bounds over all random walk steps and observe that with
probability $p' > 1 - \log^{-2} n$ a random walk does not hit any
other informed node.

To show the second part of \refLemma{claim:random-walks-do-not-hit-each-other}
we analyze the random walks phase from the following point of view. We know
that with probability at most $\log^{-2} n$ a random walk hits the informed set.
Therefore, we consider the experiments of starting one random walk after
another. Each of these trials fails with at most above probability. However,
the trials are negatively correlated, since we omit those random walks that
interfere with the informed set and thus also with another random walk. Note
that we only regard those random walks as valid that do not interfere with the
informed set at least once and only choose communication stubs that have not been
previously selected.

Since the correlation is negative we can apply Chernoff bounds on the the
number of random walks that fail. Let $X_i$ denote an indicator random variable
for the $i$-th random walk, defined as 
\begin{equation*} 
X_i = \begin{cases}
        1 & \text{if the $i$-th random walk fails} \\
        0 & \text{otherwise.}
    \end{cases}
\end{equation*}
Let furthermore $X$ be the number of random walks that fail, defined as $X =
\sum_{i=1}^{W}X_i$. From \refLemma{claim:number-of-random-walks-in-I} we obtain
that the total number of random walks visits in $I_m(t)$ is in
$\Theta(|I_m(t)|)$. The expected value of $X$ can be bounded by
$\Expected{X} \leq W\cdot \log^{-2} n = \LittleO{|I_m(t)|}$. We show that $X$ is
concentrated around its expected value as follows.
\begin{equation*}
\Probability{X \geq \left(1+\frac{1}{\log n}\right)\Expected{X}} \leq e^{-\frac{\Expected{X}}{3\log^2 n}} \leq n^{-\LittleOmega{1}}
\end{equation*}
Therefore we have only $\LittleO{|I_m(t)|}$ random walks that exhibit undesired
behavior \whp and thus the lemma holds.
\end{proof}

\begin{lemma} \label{claim:broadcasting-works}
The broadcasting procedure during the last $1/2\log \log n$ steps of a round
$r$ in Phase~\RNum{2} informs $\BigTheta{|I_m(r)| \cdot \sqrt{\log n}}$ nodes,
\whp.
\end{lemma}

\begin{proof}
Let $w$ denote an arbitrary but fixed random walk and let  $\kappa$ denote a
constant. We distinguish the following two cases
to show that the probability that a node $u_i$ opens a connection to an already
informed node can be bounded for both, sparse and dense random graphs by $\log^{-2} n$.

\textbf{Case 1: }$d \leq \log^\kappa n$.
Each random walk operates in its own \textit{safe area} as described in
\refLemma{claim:random-walks-end-in-uninformed-area}. That means, we only
consider random walks that have a distance of at least $\log\log{n}$
between each other. Therefore, in a broadcast procedure of at most
$1/2\cdot\log\log{n}$ steps no interaction between the corresponding broadcast
trees can occur.
Let $u_i$ be the $i$-th node with respect to a level-order traversal of the
message distribution tree of nodes informed by an arbitrary but fixed random walk. Let furthermore $X_i$
denote an indicator random variable for the connection opened by $u_i$ defined as 
\begin{equation*} 
	X_i = \begin{cases}
		1 & \text{if $u_i$ opens a back connection} \\
		0 & \text{otherwise.}
	\end{cases}
\end{equation*}

The claim follows from the \textit{pseudo-tree} structure of the local
subgraph, since every node has at most a constant number of edges directed
backwards and furthermore we only regard random walks in a safe area, \ie{}, random walks with a
distance of $\log\log{n}$ steps between each other. Therefore, the probability probability that the node $u_i$ opens a connection to an already
informed node can be bounded by $\BigO{1/d} \leq \log^{-2} n$.

We denote the random variable for the number of nodes that open backward
connections as $X$ and observe that $X \leq \sqrt{\log{n}}$ since the number of steps is $1/2\log\log{n}$. Using above indicator random variable we set $X =
\sum X_i$ with expected value $\Expected{X} \leq \log^{-3/2}n$.
Let $c \geq 3$ denote a constant. Since we can assume that $X$ has a binomial distribution we can
bound the probability that more than a constant number of $c$ nodes open backward
connections directly by $\Probability{X \geq c} \leq \log^{-3} n$ as follows.
\begin{align*}
&\Probability{X > c} \\
&= \sum_{i=c+1}^{\sqrt{\log{n}}} \binom{\sqrt{\log{n}}}{i} \left( \frac{1}{\log^2n} \right)^{i}  \left(1-\frac{1}{\log^2n} \right)^{\sqrt{\log{n}}-i} \\
&\leq \sum_{i=c+1}^{\sqrt{\log{n}}} \left(\frac{\sqrt{\log{n}}\cdot e}{i}\right)^i \left( \frac{1}{\log^2n} \right)^{i}  \left(1-\frac{1}{\log^2n} \right)^{\sqrt{\log{n}}-i} \\
&\leq \sum_{i=c+1}^{\sqrt{\log{n}}} \left(\frac{e}{i\log^{3/2}n}\right)^i \left(1-\frac{1}{\log^2n} \right)^{\sqrt{\log{n}}-i} \\
&\leq \left(\sqrt{\log n} - c - 1\right) \left(\frac{e}{\left(c+1\right)\log^{3/2}n}\right)^{c+1} \\
& \leq \frac{1}{\log^c{n}} \leq \log^{-3}n\enspace .
\end{align*}
In the worst case these $c$ nodes are the $c$ topmost nodes of the message
distribution tree and the corresponding branches of this tree are lost.
However, for a constant $c$ the resulting informed set is still in
$\BigTheta{\sqrt{\log{n}}}$.

\textbf{Case 2: } $d > \log^\kappa n$.
We consider the number of connection stubs that are available at an arbitrary
but fixed node $v$ and observe that the probability that $v$ opens an already
used stub can be bounded by
\begin{equation*}
\Probability{v \text{ opens a used stub to } u} \leq O(\log^{-\kappa + 2} n) \enspace ,
\end{equation*}
\ie{}, the total number of connections opened over the number of available
stubs. Furthermore, we bound the probability that the target stub belongs to a
node in the informed set as
\begin{equation*}
\Probability{u \text{ is informed}} \leq |I_m(t)| \cdot d / \left(n \left( d - \log^2{n} \right)\right) \enspace .
\end{equation*}

Therefore, the probability that either a previously used stub is opened or that
the target has already been informed can be bounded for sufficiently large $n$
by $\log^{-3} n$. We apply union bounds
and conclude that each random walk end informs a set of size
$\sqrt{\log{n}}$ after $1/2\log\log{n}$ steps with probability
$1-\log^{-2} n$.

\textbf{Both cases: }
We apply Chernoff bounds on the number of random walks that do not manage
to build up a sufficiently large informed set using broadcasting. In the first
case all random walks are clearly uncorrelated, since they live within their
own safe area. For the second case, we analyze the random walks one after
another as individual trials in our experiment. Whenever a random walk fails to
spread its message, we completely remove the entire random walk for our
analysis. We therefore have probabilities that are negatively correlated which
allows us to apply Chernoff bounds.

Let $X_i'$ denote an indicator random variable for a random walk $w_i$ 
 defined as
\begin{equation*} 
X_i' = \begin{cases}
		1 & \text{if the random walk $w_i$ fails broadcasting} \\
		0 & \text{otherwise.}
	\end{cases}
\end{equation*}
Let furthermore $X' = \sum_{i=1}^{W}X_i'$ denote the random variable for the
number of random walks that fail during the broadcasting steps with expected
value $\Expected{X'} \leq W/\log^2 n$ where $W$ is the total number of random walks. We show that $X'$ is
concentrated around the expected value using Chernoff bounds.
\begin{equation*}
\Probability{X' \geq \left(1+\frac{1}{\log n}\right)\Expected{X'}} \leq e^{-\frac{\Expected{X'}}{3\log^2 n}} \leq n^{-\LittleOmega{1}}
\end{equation*}

Since this result holds \whp, we have a set of informed nodes of size
$|I_m(r+1)| = \BigTheta{|I_m(r)| \cdot \sqrt{\log n}}$ and thus the claim holds.
\end{proof}

From \refLemma{claim:broadcasting-works} we obtain that the set of informed
vertices grows in each round by a factor of at least $\BigTheta{\sqrt{\log{n}}}$ as long
as the number of informed vertices is in $\BigO{n\cdot 2^{-\log n/\log\log
n}}$, \whp. %
Assume that the exact factor for the growth in each round is $a\sqrt{\log{n}}$ where $a$ denotes a constant. Then, the number of informed nodes that can be reached in Phase~\RNum{2} is at most
\begin{align*}
 & \left(a\sqrt{\log{n}}\right)^{4 \log n/\log \log n}\\
&= \left(a^2\right)^{2\log n/\log \log n} \cdot \left(\left(\sqrt{\log{n}}\right)^{2 \log n/\log \log n} \right)^2\\
& = \left(a^2 \cdot \sqrt{\log{n}}\right)^{2 \log n/\log \log n} \cdot \left(\sqrt{\log{n}}\right)^{2 \log n/\log \log n}\\
& \gg n \cdot 2^{-\log{n}/\log\log{n}} \enspace .
\end{align*}%
We apply a union bound over all messages and conclude we reach the bound on the number of informed vertices for Phase~\RNum{2}
after at most $4 \log n/\log \log n $ rounds \whp.

\Phase{3}{Broadcast}
In the last phase we use a simple push-pull broadcasting procedure to inform
the remaining uninformed nodes. Once $\Omega(n/2^{\log n/\log \log n})$ nodes
are informed after the second phase, within $O(\log n/\log \log n)$ additional steps at least $n/2$
nodes become informed, \whp. Furthermore, after additional $O(\log n/\log\log
n)$ steps, all nodes are informed \whp \cite{Els06}.

\begin{lemma} \label{claim:broadcast-1}
After applying push-pull for $O(\log n/\log \log n)$ steps, at most $n/\log{n}$
uninformed vertices remain for every message $m$, \whp. This procedure has a
runtime complexity in $\BigO{\log{n}/\log\log{n}}$ and an overall message complexity
in $\BigO{n\log{n}/\log\log{n}}$.
\end{lemma}

\begin{proof}
Lemma 4 from \cite{Els06} states that once the number of nodes possessing some
message $m$ is $\Omega(n/2^{\log n/\log \log n})$, then within additional
$O(\log n/\log \log n)$ steps the number of nodes informed of $m$ exceeds
$n/2$.\footnote{The two lemmas are stated w.r.t.~Erd\H{o}s-R\'enyi graphs. The same proofs, however, lead to the same statements for the configuration model, too.}
We observe that the number of informed nodes is within the bounds
required in Lemma 4 from \cite{Els06} for each message. We conclude that the
set of informed vertices underlies an exponential growth \whp. Therefore,
$|I_m(t)| \geq n/2$ after additional $\BigO{\log n/\log\log n}$ steps, using
$\BigO{n \log n/\log\log n}$ messages.
Furthermore, we apply Lemma 5 from \cite{Els06}, which states that after
additional $\BigO{\log\log{n}}$ steps it holds that for the uninformed set
$|H(t)| \leq n/\log{n}$ \whp.\footnotemark[\value{footnote}]
Since both, Lemma 4 and Lemma 5 from \cite{Els06} hold \whp $1 -
\LittleO{n^{-2}}$ we use union bound over all messages and conclude that these
results hold for all messages \whp.
\end{proof}

\begin{lemma} \label{claim:broadcast-2}
After $\BigO{\log{n}/\log\log{n}}$ steps, every remaining uninformed node is
informed of message $m$ \whp.
\end{lemma}

The proof of \refLemma{claim:broadcast-2} is similar to Lemma 5 and Lemma 6
from \cite{Els06}. Our adapted version is as follows.

\begin{proof}
After performing the mixing steps during the random walk phase, we can assume
that each message is distributed uniformly at random nodes. From
\refLemma{claim:broadcast-1} we deduce that each node opens at most a number of
connections in $\BigO{\log{n}/\log\log{n}}$ after the last mixing phase,
whereas each node has at least $\log^{2+\epsilon}n$ communication stubs.
Additionally, we consider in the following phase $\BigO{\log{n}/\log\log{n}}$
\texttt{pull} steps. Each node can open up to $\BigO{\log{n}/\log\log{n}}$
additional connections during this phase and incoming connections from
uninformed nodes can be bounded by the same expression following a
balls-into-bins argument. We denote the number of opened stubs as $S$ with $S =
\BigO{\log{n}}$ and conclude that we still have at least $\log^{2+\epsilon}{n}
- S = \BigOmega{\log^{2+\epsilon}n}$ \textit{free} connection stubs available
which are not correlated to the message distribution process of message $m$ in
any way.

In each step, a node opens a connection to a randomly chosen neighbor and
therefore chooses a \textit{wasted} communication stub with probability at most
$ a\log{n}/(d \cdot \log\log{n})$ where $a$ is a constant. 

If a free stub is chosen, the corresponding communication partner is informed
of message $m$ with probability at least $|I_m(t)|\cdot(d-S)/(n\cdot d)$.
Therefore, any uninformed node $v$ \emph{remains possibly uninformed}, \ie{},
either uses an already wasted communication stub or connects to an uninformed
partner, with the following probability. 
\begin{align*}
p' &= \Probability{v \text{ remains possibly uninformed}}\\
& = 1 - \Probability{v \text{ is definitely informed}} \\
& \leq 1 - \Probability{v \text{ chooses a free stub to } u} \cdot \Probability{u \in I_m(t)} \\
& \leq 1 - \left( 1 - \frac{c}{\log{n}\cdot\log\log{n}} \right) \cdot \left(\left(1 - \frac{|H_m(t)|}{n} \right)\frac{d-S}{d} \right) \\
\intertext{We apply $H_m(t) \leq n/\log{n}$ and obtain} 
p' & \leq 1 - \left( 1 - \frac{c}{\log{n}\cdot\log\log{n}} \right) \cdot \left(\left(1 - 1/\log{n}\right)\frac{d-S}{d} \right) \\
& \leq \frac{c}{\log{n}\cdot\log\log{n}} + \left(1 - \frac{c}{\log{n}\cdot\log\log{n}}\right) \frac{d-S}{\log{n}\cdot d} \\
& < \log^{-c}{n}
\end{align*}
for a suitable constant $c$. Therefore, the probability that an arbitrary node
remains uninformed after $4\log{n}/(c\cdot\log\log{n})$ steps can be bounded
by
\begin{equation*} 
\Probability{v \text{ remains uninformed}} \leq \left(\frac{1}{\log{n}}\right)^{\frac{4\log{n}}{\log\log{n}}} = \frac{1}{n^4} \enspace . \qedhere
\end{equation*}
\end{proof}
\begin{lemma} \label{claim:broadcast-3}
After the broadcast phase, every node is informed of every message \whp.
\end{lemma}
\begin{proof}
We use union bound on the results of \refLemma{claim:broadcast-2} over all $n$
messages and over all $n$ nodes. Thus after the broadcast phase each node is
informed of every message with probability at least $1 - n^{-2}$.
\end{proof}

\begin{proof}[Proof of Theorem \ref{thm:theorem1}]
Theorem \ref{thm:theorem1} follows from the proofs of the correctness of the individual phases,
\refLemma{claim:distributionA-4} for the distribution phase,
Lemmas \ref{claim:random-walks-end-in-uninformed-area},
\ref{claim:random-walks-do-not-hit-each-other}, and
\ref{claim:broadcasting-works} for the random walks phase, and
\refLemma{claim:broadcast-3} for the broadcast phase.
\end{proof}

\section{Memory Model} \label{sect:memory}

In this section we consider the $G(n,p)$ graph, in which an edge between two
nodes exists with probability $p$, independently, and assume that the nodes
have a constant size memory. That is, the nodes can store up to four different
links they called on in the past, and they are also able to avoid these links
as well as to reuse them in a certain time step. More formally, we assume that
each node $v \in V$ has a list $l_v$ of length four. The entry $l_v[i]$
contains a link address which is connected on the other end to a fixed node
$u$. Whenever node $v$ calls on $l_v[i]$ in a step, it opens a communication
channel to $u$. From now on, we will not distinguish between the address stored
in $l_v[i]$ and the node $u$ associated with this address. As assumed in the
previous sections, such a channel can be used for bi-directional communication
in that step. Furthermore, $v$ is also able to avoid the addresses stored in
$l_v$, by calling on a neighbor chosen uniformly at random from $N(v) \setminus
\cup_{i=0}^3 \{ l_v[i]\}$, where $N(v)$ denotes the set of neighbors of $v$.
This additional operation is denoted \texttt{open-avoid} in
\autoref{alg:algorithm-3} and \autoref{alg:leader-election}. Note that the approach of
avoiding a few previously contacted neighbors was also considered in the
analysis of the communication overhead produced by randomized broadcasting
\cite{BEF08,ES08} and in the analysis of the running time of push-pull protocols
in the preferential attachment model \cite{DFF11}. Clearly, the list $l_v$ may
also be empty, or contain less than $4$ addresses.

The algorithm we develop is similar to the one addressed in \cite{BCEG10} for
complete graphs. However, there are two main differences. While in
\cite{BCEG10} the protocol just uses the fact that in the random phone call
model the nodes of a complete graph do not contact the same neighbor twice
\whp, this cannot be assumed here. Furthermore, to obtain a communication
infrastructure for gathering information at a so-called leader, we use some
specific structural property of random graphs which was not necessary in
complete graphs. There, we built an infrastructure by using communication paths
in an efficient broadcasting network obtained by performing broadcasting once.
Here, we need to analyze the structure of random graphs in relation with the
behavior of our algorithm.

\begin{algorithm}
\SetAlgoVlined
\DontPrintSemicolon
\SetKwFunction{Push}{push}
\SetKwFunction{Pull}{pull}
\SetKwFunction{Open}{open}
\SetKwFunction{OpenAvoid}{open-avoid}
\SetKwFunction{Close}{close}
\SetKw{KwLet}{Let}
\SetKw{KwRun}{Run}
\SetKwFor{At}{at each}{do in parallel}{end}
\SetKwFor{ForEach}{for each}{do}{end}
\SetKwFor{With}{with probability}{do}{end}
\SetKwFor{ForRoundNode}{for round}{do at every node $v$ in parallel}{end}
\SetKwFor{For}{for}{do}{end}
\SetKwFor{ForRound}{for round}{do}{end}
Assume a leader is given.\\
\textbf{Phase \RNum{1}}\\
\For{$t = 0$ \KwTo $3$}{
	The leader performs an \OpenAvoid and then then a \Push{$m_v(0)$} operation. In each step, 
the leader stores in $l_v[t]$ the address of the node contacted in this step.
}
\For{$t = 4$ \KwTo $4 \log_4 n + 4\rho\log \log n$}{
	Every node $v$ that received $m_v(0)$  in step $t$ {\em for the first time} (with $t=4j+k$ and $k\in \{0,1,2,3\}$) is active in  step $4(j+1), 4(j+1)+1$, $4(j+1)+2$, and $4(j+1)+3$.\;
	Every active node $v$ performs an \OpenAvoid and then a \Push{$m_v(0)$} operation. $v$ 
stores in $l_v[t \mbox{ mod } 4]$ the address of the node contacted in the current step. \;
	Every active node $v$ also stores the time steps $4(j+1)$, $4(j+1)+1$, $4(j+1)+2$ and $4(j+1)+3$ together with the neighbors it used for the \Push operations in the list $l_v$.\;
}
\For{$t=4\log_4 n + 4\rho\log \log n+1$ \KwTo $4\log_4 n + 8\rho \log \log n$}{
Every node $v$ that knows $m_v(0)$ performs \Pull{$m_v(0)$} operation.\;
Every node $v$ that does not know $m_v(0)$ performs an \OpenAvoid and receives eventually ($m_v(0)$). 
The address of the contacted node is stored in $l_v[t \mbox{ mod } 4]$.\;
Every node $v$ that receives $m_v(0)$ for the first time in step $t$ remembers the chosen
neighbor together with $t$ in the list $l_v[0]$.\;
}
\AlgoPhase{2}
$t'\gets4\log_4 n + 8\rho \log \log n$\;
\For{$t=1$ \KwTo $ \rho \log \log n$}{
Every node $v$ which received the message in step $t'-t+1$ (for the first time) opens a channel 
to the corresponding neighbor in
$l_v[0]$ and performs a \Push operation with all original messages it has.\;
}
$t'\gets4\log_4 n + 8\rho \log \log n$\;
\For{$t=1$ \KwTo $4\log_4 n + 8 \rho \log \log n$}{
Every node $v$ which stores a neighbor with time step $t'-t+1$ in its list $l_v$ opens
a channel to that neighbor in  $l_v$ and receives the message from that neighbor. The node at the
other side performs a \Pull operation with all original messages it has.\;
}
\AlgoPhase{3}
The leader broadcasts all original messages using the algorithm described in Phase I for message $m_v(0)$.\;
\vspace{0.5\baselineskip}
\caption{Gossiping algorithm. After each step, the nodes close all channels opened in that 
step.}\label{alg:algorithm-3}
\end{algorithm}

\subsection{Leader Election}
\label{sect:leader-election}

\begin{algorithm}
\SetAlgoVlined

\SetKwFunction{Push}{push}
\SetKwFunction{Pull}{pull}
\SetKwFunction{Open}{open}
\SetKwFunction{OpenAvoid}{open-avoid}
\SetKwFunction{Close}{close}
\SetKw{KwLet}{Let}
\SetKwFor{AtNode}{at each node}{do in parallel}{end}
\SetKwFor{ForEach}{for each}{do}{end}
\SetKwFor{With}{with probability}{do}{end}
\SetKwFor{For}{for}{do}{end}
\SetKwFor{ForRound}{for round}{do}{end}

\AtNode{$v$}{
	\With{$\log^2{n} / n$}{
		$v$ becomes active\;
		\OpenAvoid{};
		\Push{$I\!D_v$}\;
	}
}
\For{$t = 1$ \KwTo $\log{n} + \rho\log\log n$}{
	\AtNode{$v$}{
		\If{$v$ has incoming messages $m$}{
			$v$ becomes active\;
		}
		\KwLet $i_v(t)$ be the smallest identifier that $v$ received so far\;
		\If{$v$ is active}{
			\OpenAvoid{};
			\Push{$i_v$}\;
		}
	}
}
\For{$t = 1$ \KwTo $\rho\log\log n$}{
	\AtNode{$v$}{
		\OpenAvoid{};
		$i_v \gets \min\{i_v, \Pull{}\}$\;
	}
}
\AtNode{$v$}{
	\If{$I\!D_v = i_v$}{
		$v$ becomes the leader\;
	}
}

\caption{Leader Election Algorithm. After each step, the nodes close all channels opened in
that step. \label{alg:leader-election}}
\end{algorithm}

In our main algorithm, we assume that a single node is aware of its role as a leader.
The other nodes, however, do not necessarily have to know the ID of this node. They just have 
to be aware of the fact that they are not leaders.
In order to find a leader we may apply the following leader election algorithm described in
\autoref{alg:leader-election} (cf.~\cite{BCEG10}). Each node flips a coin, and
with probability $\log^2 n/n$ it becomes a \emph{possible} leader. We assume
that every node $v$ has a unique ID denoted by $I\!D_v$. Each possible leader
starts a broadcast, by sending its ID to some nodes chosen uniformly at random
from the set of its neighbors, except the ones called in the previous three
steps. Once a node receives some ID, it becomes active, and starts propagating
the smallest ID it received so far. This push phase is performed for $\log n +
\rho \log \log n$ steps, where $\rho > 64$ is some large constant. In the last $\rho
\log \log n$ steps, the IDs of the possible leaders are spread by pull
transmissions. The possible leader with the smallest ID will become the leader.

\begin{lemma} \label{lem_leader}
At the end of \autoref{alg:leader-election}, all nodes are aware of the leader,
\whp.
\end{lemma}

\begin{proof}
Let us denote by $I(t)$ the set of nodes at time $t$, which have received some
ID by this time step. Lemma 2.2 of \cite{ES08} states that a message is
distributed by a modified push-pull algorithm, in which each node is allowed to
avoid the $3$ neighbors chosen in the previous $3$ steps, is distributed to
$n-n/\sqrt[4]{n}$ nodes in $\log n + \rho \log \log n$ steps\footnote{We
adapted the running time from the lemma mentioned before to our algorithm.}.
This implies that by this time $I(t) \geq n - n/\sqrt[4]{d}$, and the number of
message transmissions is at most $O(n \log \log n)$, \whp. Furthermore,
$n/\log^{O(1)} n$ nodes know the leader. According to Lemma 2.7.~and 2.8.~from
\cite{ES08}, after additional $O(\log \log n)$ steps, the message is
distributed to all nodes, \whp. This implies that after this number of
additional steps $I(t) =n$, \whp, and $\Omega(n/\log^2 n)$ nodes know the
leader, \whp. Applying Lemmas 2.7.~and 2.8.~from \cite{ES08} again, we obtain
the lemma.
\end{proof}

Now we consider the robustness of the leader election algorithm. We show that by applying our
algorithm, one can tolerate up to $n^{\epsilon'}$ random
node failures, \whp, where $\epsilon' < 1/4$ is a small constant. That is, during the 
execution of the algorithm, $n^{\epsilon'}$ nodes, chosen uniformly and independently at random,
may fail at any time. The node failures are non-malicious, i.e., a failed node does not communicate 
at all. The theorem below
is stated for $p = \log^5 n/n$. However, with an extended analysis, the theorem
can be generalized to any $p > \log^{2+\epsilon} n/n$. 

\begin{lemma} \label{lem_leader2}
In the failure model described above, at the end of \autoref{alg:leader-election} the leader is aware of its role,
and all other nodes know that they are not the leader
\whp.
\end{lemma}
\begin{proof}
Here we only consider the node, which decided to become a possible leader, and 
has the smallest ID among such nodes. The algorithm is the same as the sequential 
version of the broadcast algorithm given in \cite{ES08}. We know that within the 
first $(1-\epsilon') \log n - \rho \log \log n$ steps, the number of informed nodes (i.e., 
the number of nodes receiving the ID of the node we consider) is
$n^{1-\epsilon'}/\log^{2+\Omega(1)} n$, \whp. Since 
$n^{\epsilon'}/n$ nodes may fail in total, independently, Chernoff bounds imply that 
all nodes informed within the first $(1-\epsilon') \log n - \rho \log \log n$ steps
are healthy, \whp. We also know that after $\log n + \rho \log \log n$ push steps, the 
number of informed nodes is $n-n/\log^{4+\Omega(1)} n$, \whp \cite{ES08}. On the other side, 
if we only consider push transmissions, the number of nodes which become informed by a message 
originating from a node informed after step $(1-\epsilon') \log n - \rho \log \log n$ is at most 
$2^{\epsilon' + 2\rho \log \log n} = n^{\epsilon'} \log^{O(1)} n$. This is due to the fact that the number of 
informed nodes can at most double in each step. Thus, the total number of nodes, which received the message 
from a failed node in the first $\log n + \rho \log \log n$ push steps, if this node would not fail, 
is at most $n^{2\epsilon'} \log^{O(1)} n$. The probability that 
one of the possible leaders is not among the nodes informed in the first $\log n + \rho \log \log n$ push steps,
or is not informed due to a node failure, is $o(\log^{-4} n)$. The union bound over $O(\log^2 n)$ 
possible leaders implies the lemma.
\end{proof}

\subsection{Gossiping Algorithm and its Analysis}

The pseudocode can be found in \autoref{alg:algorithm-3}. 
We assume that at the beginning a random node acts as a leader. For an efficient and 
robust leader election algorithm see \autoref{alg:leader-election}.
Once a leader is
given, the goal is to gather all the messages at this node. First, we build an
infrastructure as follows (Phase I). The leader emits a message by contacting
four different nodes (one after the other), and sending them these messages.
These nodes contact four different neighbors each, and send them the message.
If we group four steps to one so-called long-step, then in long-step $i$, each
of the nodes which received the message in the long-step before \emph{for the
first time} chooses four distinct neighbors, and sends them the message.
Furthermore, each node stores the addresses of the chosen nodes. This is
performed for $\log_4 n + \rho \log \log n$ long-steps, where 
$\rho >64$ is some large constant. For the next $\rho \log
\log n$ long-steps, all nodes, which have not received the message of the leader so far,
choose $4$ different neighbors in each of these long-steps, and open
communication channels to these nodes (i.e., communication channels are opened
to all these different neighbors within one long-step, where each of these neighbors
is called in exactly one step). If some node has the message of the leader
in some step, then it sends this message through the incident communication
channel(s) opened in that step. We call these last $\rho \log \log n$
long-steps \emph{pull long-steps}. 

In Phase II the infrastructure built in Phase I is used to send the message of
each node to the leader. This is done by using the path, on which the the
leader's message went to some node, to send the message of that node back to
the leader. In the third phase the messages gathered by the leader are sent to
all nodes the same way the leader's message was distributed in Phase I. Then,
the following lemmas hold.

\begin{lemma} \label{lem_memory_I}
After $\log_4 n + \rho \log \log n$ long-steps at least $n/2$ nodes have the
message of the leader, \whp. \seeProof{appendix:omitted-proofs-2}
\end{lemma}

\begin{proof}
Since during the whole process every node only chooses four neighbors, simple
balls-into-bins arguments imply that the total number of incoming communication
channels opened to some node $u$ is $O(\log n)$, with probability at least
$1-n^{-4}$ \cite{RS98}. 

Let $v$ be the leader, and let its message be $m_v(0)$. We know that as long as
$d = 2^{o(\sqrt{\log n})}$, the tree spanned by the vertices at distance at
most $\rho \log \log n$ from $v$ is a tree, or there are at most $4$ edges
which violate the tree property \cite{BES14}. Thus, after $\rho \log \log n$
steps, at least $3^{\rho \log \log n -1}$ vertices have $m_v(0)$ \whp. If $d =
2^{\Omega(\sqrt{\log n})}$ simple probabilistic arguments imply that at least
$3^{\rho \log \log n -1}$ vertices have $m_v(0)$ \whp. 

Let now $I^+(t)$ be the set of nodes, which receive $m_v(0)$ in long-step $t$
(for the first time). Each of these nodes chooses an edge, which has already
been used (as incoming edge), with probability $O(\log n/d) \leq
1/\log^{1+\epsilon/2} n$. Let $|I(t)| \leq n/\log^2 n$ and $I^+ (t) = \{ v_1,
\dots , v_{|I^+(t)|}\}$. Given that some $v_i$ has at least $pn(1-o(1))$
neighbors in $G$, and at most $|I(t)| + 4|I^+(t)|)p(1 +o(1)) + 5\log n$
neighbors in $I(t) \cup \{v_1, \dots , v_{i-1}\}$, the edge chosen by $v_i$ in
a step of the long-step $t+1$ is connected to a node, which is in $I(t)$ or it
has been contacted by some node $v_1, \dots , v_{i-1}$ in long-step $t+1$, with
probability at most 
\begin{equation} \label{equI}
p_I \leq \frac{(|I(t)| + 4|I^+(t)|)p(1 +o(1)) + 5\log n}{pn(1 - o(1))}
\end{equation} 
independently (cf.~\cite{Els06}). Thus, we apply Chernoff bounds, and obtain
that the number of newly informed nodes is 
$$|I^+(t+1) \geq 4|I^+(t)|\left(1-\frac{2}{\log^{1+\epsilon/2} n}\right)\enspace,$$
with probability $1-n^{-3}$. Therefore, after $\log_4 n - O(\log \log n)$
steps, the number of informed nodes is larger than $n/\log^2 n$.

Now we show that within $\rho \log \log n$ steps, the number of uninformed
nodes becomes less than $n/2$. As long as $|I(t)| \leq n/3$, applying equation
(\ref{equI}) together with standard Chernoff bounds as in the previous case, we
obtain that 
$$|I^+(t+1) \geq 4|I^+(t)| - |I^+(t)|\frac{5|I^+(t)(1+o(1))}{n} > 2|I^+(t)|\enspace,$$
with probability $1-n^{-3}$. Once $|I(t)|$ becomes lager than $n/3$, it still
holds that $|I^+(t)| \geq |I(t)|$ (see above). Thus, in the next step the total
number of informed nodes exceeds $n/2$, \whp.
\end{proof}

The approach we use here is similar to the one used in the proof of Lemma
2.2.~in \cite{ES08}; the only difference is that in \cite{ES08} the nodes
transmitted the message in all steps, while here each node only transmits the
message to $4$ different neighbors chosen uniformly at random. Note that each
node only opens a channel four times during these $\log_4 n + \rho \log \log n$
long-steps, which implies a message complexity of $O(n)$.

\begin{lemma} \label{lem_memory_II}
After $\rho \log \log n$ pull long-steps, all nodes have the message of the
leader \whp. \seeProof{appendix:omitted-proofs-2}
\end{lemma}

\begin{proof}
First we show that within $\rho \log \log n/2$ steps, the number of uninformed
nodes decreases below $n/\sqrt[4]{d}$. The arguments are based on the proof of
Lemma 2.2. from \cite{ES08}. Let us consider some time step $t$ among these
$\rho \log \log n/2$ steps. Given that all nodes have some degree $\Omega(d)$,
a node chooses an incident edge not used so far (neither as outgoing nor as
incoming edge) with probability $1-O(\log n/d)$. According to Lemma 1 of
\cite{Els06}, this edge is connected to a node in $H(t)$ with probability at
most $O(p|H(t)|+ \log n)/d)$, independently. Applying Chernoff bounds, we
obtain that as long as $|H(t)| > n/\sqrt[4]{d}$, we have $$|H(t+1)| \leq
O\left( |H(t)| \cdot \left( \frac{|H(t)|}{n} + \frac{\log n}{d}\right) \right)
\enspace, $$ \whp. Thus, after $\rho \log \log n/2$ steps, the number of
uninformed nodes decreases below $n/\sqrt[4]{d}$, \whp (cf.~\cite{KSSV00}).
Applying now Lemmas 2.7. and 2.8. from \cite{ES08} (for the statement of these
lemmas see previous proofs), we obtain the lemma. Since only nodes of $H(t)$
open communication channels in a step, we obtain that the communication
complexity produced during these pull long-steps is $O(n)$, \whp.
\end{proof}

\begin{lemma} \label{lem_phase2}
After Phase II, the leader is aware of all messages in the network, \whp.
\seeProof{appendix:omitted-proofs-2} 
\end{lemma}

\begin{proof}
Let $w$ be some node, and we show by induction that the leader receives
$m_w(0)$. Let $t$ be the long-step, in which $w$ receives $m_v(0)$. If $t =1$,
then $w$ is connected to $v$ in the communication tree rooted at $v$, and $v$
receives $m_w(0)$ in one of the last four steps of Phase II.

If $t > 1$, then let $w'$ denote the successor of $w$ in the communication tree
rooted at $v$. That is, $w$ received $m_v(0)$ from $w'$ in long-step $t$. This
implies that $w'$ either received $m_v(0)$ in pull long-step $t$ or $t-1$, or
it received $m_v(0)$ in a push long-step $t-1$. If however, $w'$ obtained
$m_v(0)$ in pull long-step $t$, then this happened before $w$ received the
message. In both cases $w'$ will forward $m_w(0)$ to $v$ together with
$m_{w'}(0)$, according to our induction hypothesis, and the lemma follows.
\end{proof}

\begin{lemma} \label{lem_phase3}
After Phase III, gossiping is completed \whp. 
\end{lemma}

The proof of \autoref{lem_phase3} follows directly from \autoref{lem_memory_II}. From the lemmas
above, we obtain the following theorem.

\begin{theorem} \label{theo_memory}
\Whp \autoref{alg:algorithm-3} completes gossiping in $O(\log n)$ time steps by
producing $O(n)$ message transmissions. If leader election has to be 
applied at the beginning, then the 
communication complexity is $O(n \log \log n)$.
\end{theorem}

Now we consider the robustness of our algorithm. We show that by applying our
algorithm twice, independently, one can tolerate up to $n^{\epsilon'}$ random
node failures, \whp, where $\epsilon'<1/4$. That is, during the 
execution of the algorithm, $n^{\epsilon'}$ nodes, chosen uniformly and independently at random,
may fail at any time. The node failures are non-malicious, i.e., a failed node does not communicate 
at all. The theorem below
is stated for $p = \log^5 n/n$. However, with an extended analysis, the theorem
can be generalized to any $p > \log^{2+\epsilon} n/n$. As before, we assume 
that a random node acts as a 
leader. Since at most $n^{\epsilon'}$ random nodes fail in total, the leader fails during the 
execution of the algorithm with probability $n^{-\Omega(1)}$. 
Moreover, due to the 
robustness of the leader election algorithm from \autoref{sect:leader-election}, the result of
\autoref{theo_memory_II} 
also holds if leader election has to be applied to find a leader
at the beginning.

\begin{theorem} \label{theo_memory_II}
Consider a $G(n,p)$ graph with $p = \log^{5} n/n$. Assume that $f=
n^{\epsilon'}$ random nodes fail according to the failure model 
described above, where $\epsilon' < 1/4$. If we run
\autoref{alg:algorithm-3} two times, independently, then at the end $n-|f|(1+o(1))$ nodes know
all messages of each other, \whp. \seeProof{appendix:omitted-proofs-2}
\end{theorem}

\begin{proof}
To analyze the process, we assume that all the failed nodes fail directly after
Phase I and before Phase II. This ensures that they are recorded as
communication partners for a number of nodes, but these failed nodes are not
able to forward a number of messages in Phase II to the leader. Let us denote
the two trees, which are constructed in Phase I of the two runs of the
algorithm, by $T_1$ and $T_2$, respectively. First we show that with
probability $1-o(1)$ there is no path from a failed node to another failed node
in any of these trees. Let us first build one of the trees, say $T_1$.
Obviously, at distance at most $(1-\epsilon')\log_4 n - \rho \log \log n$ from
the root, there will be less than $n^{1-\epsilon'}/\log^2 n$ nodes. Thus, with
probability $(1-n^{\epsilon'}/n)^{n^{1-\epsilon'}/\log^2 n} = 1-o(1)$, no node
will fail among these $n^{1-\epsilon'}/\log^2 n$ many nodes. This implies that
all the descendants of a failed node will have a distance of at most $\epsilon'
\log_4 n + O(\log \log n)$ to this node. Then, the number of descendants of a
failed node is $n^{\epsilon'} \log^{O(\log \log n)} n$, given that the largest
degree is $\log^{O(1)} n$. As above, we obtain that none of the failed nodes is
a descendant of another failed node with probability $1-o(1)$. 

For simplicity we assume that each failed node participates in at least one
push long-step, i.e., it contacts $4$ neighbors and forwards the message of the
leader (of $T_1$ and $T_2$, respectively) to these neighbors. Now we consider
the following process. For each failed node $v$, we run the push-phase for
$\epsilon' \log_4 n + O(\log \log n)$ long-steps. The other nodes do not
participate in this push phase, unless they are descendants of such a failed
node during these $\epsilon' \log_4 n + O(\log \log n)$ long-steps. That is, if
a node is contacted in some long-step $i$, then it will contact $4$ neighbors
in long-step $i+1$; in long-step $1$ only the failed nodes are allowed to
contact $4$ neighbors. Then, we add to each node $w \neq v$ in the generated
tree rooted at $v$ all nodes being at distance at most $\rho \log \log n$ from
$w$. Clearly, the number of nodes in such a tree rooted at $v$ together with
all the nodes added to it is $n^{\epsilon'} \log^{O(\log \log n)} n$. This is
then repeated a second time. The nodes attached to $v$ in the first run are
called the descendants of $v$ in $T_1$ in the following. Accordingly, the
corresponding nodes in the second run are called the descendants of $v$ in
$T_2$.

We consider now two cases. In the first case, let $v$ be a failed node, and
assume that $v$ contacts four neighbors in $T_2$, which have not been contacted
by $v$ in $T_1$. Such a failed node is called friendly. Furthermore, let
$F(T_1)$ be the set of nodes which are either failed or descendants of a failed
node in $T_1$. As shown above, $|F(T_1)| = n^{\epsilon'} \cdot n^{\epsilon'}
\log^{O(\log \log n)} n$, with probability $1-o(1)$. Let $v_i$, $i = 1, 2,
\dots $ be the descendants of $v$ in $T_2$, and denote by $A_i$ the event that
$v_i \not\in F(T_1)$. Then, 
\begin{align*}
\Pr[{\overline{A_i}} ~|~ A_1 \dots A_{i-1}] &\leq \frac{\log^{5} n \cdot n^{2\epsilon'} \log^{O(\log \log n)} n}{n} \\
&< \frac{n^{2\epsilon'} \log^{O(\log \log n)} n}{n} \enspace .
\end{align*}
Since $v$ has at most $n^{\epsilon'} \log^{O(\log \log n)} n$ descendants, none
of them belongs to $F(T_1)$, with probability at most $n^{3\epsilon'}
\log^{O(\log \log n)} n /n$. Thus, the expected number of descendants of
friendly failed nodes in $F_1 \cap F_2$, is $o(1)$, as long as $\epsilon' <
1/4$. 

In the second case, we denote by $NF_1$ the set of nodes, which are direct
descendants of non-friendly failed nodes. That is, $NF_1$ are the nodes which
are contacted by non-friendly failed nodes in step $1$ of the process described
above. Since $p= \log^5 n/n$, a failed node is non-friendly with probability
$O(1/\log^5 n)$. Using standard Chernoff bounds, we have $|NF_1| = O(|f|/\log^5
n)$, \whp. Let now $NF_2$ denote the set of nodes which are either contacted by
the nodes $NF_1$ in a push long-step of the original process in $T_1$ as well
as in $T_2$, or contact a node in $NF_1$ in a pull long-step of the original
process in both, $T_1$ and $T_2$. Similarly, $NF_{i+1}$ denotes the set of
nodes which are either contacted by the nodes $NF_i$ in a push long-step of the
original process in $T_1$ as well as in $T_2$, or contact a node in $NF_i$ in a
pull long-step of the original process in both $T_1$ and $T_2$. We show that
$|NF_{i+1}| < |NF_i|$ \whp, and for any non-friendly failed node $v$ there is
no descendant of $v$ in $NF_i$ with probability $1-o(1)$. The second result
implies that $|NF_{\rho \log n}| =0$ \whp, if $\rho$ is large enough. The first
result implies then that $\sum_{i=1}^{\rho \log n} |NF_i| < O(|f| / \log^3 n)$,
\whp. 

To show the first result we compute the expected value $E[NF_{i+1}]$ given
$NF_i$. Clearly, a node contacted by a node of $NF_i$ in $T_1$ is contacted in
$T_2$ as well with probability $O(1/\log^5 n)$. Similarly, a node which
contacted a node of $NF_i$ in $T_1$ contacts the same node in $T_2$ with
probability $O(1/\log^5 n)$. Simple balls into bins arguments imply that the
number of nodes, which may contact the same node, is at most $O(\log n/\log
\log n)$ \cite{HR90}. Applying now the method of bounded differences, we have
$|NF_{i+1}| < |NF_i|$ \whp, as long as $NF_i$ is large enough.

The arguments above imply that if the number of descendants of a non-friendly
failed node in $NF_i$ is at least $\rho \log n$ for some $\rho$ large enough,
then the number of descendants in $NF_{i+1}$ does not increase, \whp.
Furthermore, as long as the number of these descendants $NF_i$ is $O(\log n)$,
then there will be no descendants in $NF_{i+1}$ with probability $1-o(1)$, and
the statement follows. Summarizing, $$\sum_{i=1}^{\infty} NF_i = O(|f|/\log^3
n)$$ \whp, which concludes the proof.
\end{proof}

\section{Empirical Analysis}

We implemented our algorithms using the \Cpp~programming language and ran
simulations for various graph sizes and node failure probabilities using four
64 core machines equipped with 512~GB to 1~TB memory running on Linux. The
underlying communication network was implemented as an Erd\H{o}s-R\'enyi random
graph with $p = \log^2{n}/n$. We measured the number of steps, the average
number of messages sent per node, and the robustness of our algorithms. The
main result from \autoref{sect:main-result} shows that it is possible to reduce
the number of messages sent per node by increasing the running time. This effect
can also be observed in \autoref{plot:comparison}, where the communication
overhead of three different methods is compared. The plot shows the average
number of messages sent per node using a simple push-pull-approach,
\autoref{alg:fast-gossiping}, and
\autoref{alg:algorithm-3}. In the simple push-pull-approach, every node opens 
in each step a communication channel to a randomly selected neighbor, and each 
node transmits all its messages through all open channels incident to it. This is 
done until all nodes receive all initial messages.


\begin{figure}
\centering
\includegraphics[scale=1.5]{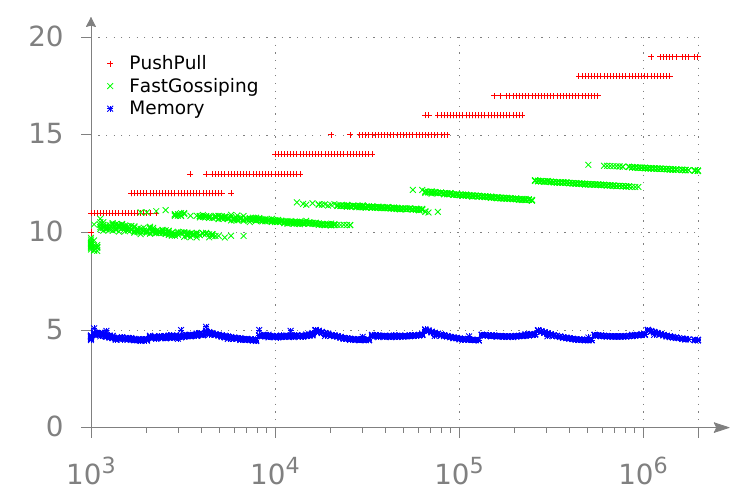}
\caption{Comparison of the communication overhead of the gossiping methods. The
$x$-axis shows the graph size, the $y$-axis the average number of messages sent
per node.}
\label{plot:comparison}
\end{figure}
\begin{figure}
\centering
\includegraphics[scale=1.5]{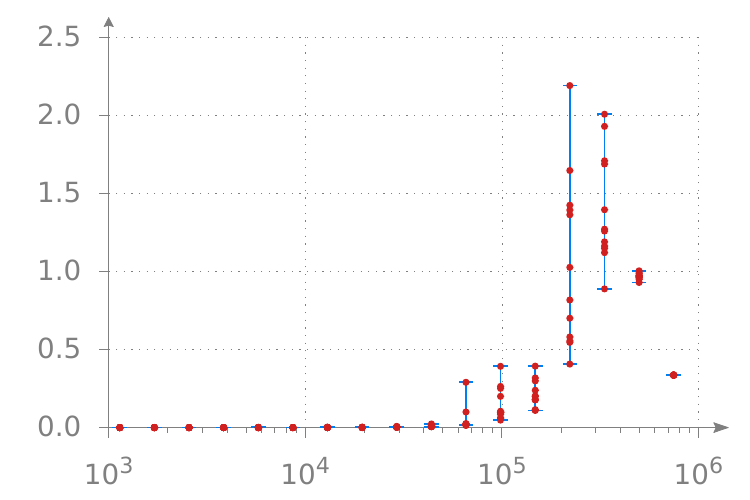}
\caption{Relative number of additional node failures in the memory model with a
graph size of 1,000,000. The $x$-axis shows the number of nodes marked \emph{failed} $F$,
the $y$-axis the ratio of additional uninformed nodes to $F$.}
\label{plot:robustness}
\end{figure}

\autoref{plot:comparison} shows an increasing gap between the message
complexity of \autoref{alg:fast-gossiping} and the simple push-pull
approach. Furthermore, the data shows that the number of messages sent per node
in \autoref{alg:algorithm-3} is bounded by $5$. According to the descriptions
of the algorithms, each phase runs for a certain number of steps. The
parameters were tuned as described in \autoref{tab:simulation-limits} to obtain
meaningful results.


\begin{table}
\centering
\begin{tabularx}{\columnwidth}{l|X|l}Phase&Limit&Value\\\hline
\multicolumn{3}{c}{ \autoref{alg:fast-gossiping} } \\
\hline I & number of steps & $\ceil{1.2\cdot\log\log{n}}$ \\
II & number of rounds & $\ceil{\log{n} / \log\log{n}}$ \\
II & random walk probability & $1.0 / \log{n}$ \\
II & number of random walk steps & $\ceil{\log{n} / \log\log{n} + 2}$ \\
II & number of broadcast steps & $\ceil{0.5\cdot\log\log{n}}$ \\
\hline \multicolumn{3}{c}{ \autoref{alg:algorithm-3} } \\
\hline I & first loop, number of steps \newline (rounded to a multiple of $4$) & $2.0 \cdot \log{n}$ \\
I & second loop, number of steps & $\floor{2.0 \cdot \log\log{n}}$ \\
II & number of steps & corresponds to Phase I \\
III & number of push steps & $\floor{\log{n}}$ \\
\end{tabularx}
\caption{The actual constants used in our simulation.}
\label{tab:simulation-limits}
\end{table}

The fact that the number of steps is a discrete value also explains the
discontinuities that can be observed in the plot. In the case of the simple
push-pull-approach, these jumps clearly happen whenever an additional
step is required to finish the procedure. Note, that since in this approach
each node communicates in every round, the number of messages per node
corresponds to the number of rounds.

In the case of
\autoref{alg:fast-gossiping}, we do not only observe these jumps, but also a
reduction of the number of messages per node between the jumps.
Let us consider such a set of test runs between two jumps. 
Within such an interval, the number of
random walk steps as well as broadcasting steps remain the same while $n$ increases. 
The number of random walks, however, is not fixed. Since each node starts a
random walk with a probability of $1/\log{n}$, the \emph{relative}
number of random walks decreases and thus also the average number of messages per node (see also \autoref{plot:comparison2} in
Appendix \ref{appendix:plots}).
This shows the impact of the random walk phase on the message
complexity.

The last phase of each algorithm was run until the entire graph was
informed, even though 
the nodes do not have this
type of global knowledge. 
From our
data we observe that the resulting number of steps is 
concentrated (i.e., for the same $n$ the number of steps to complete 
only differs by at most $1$ throughout all the simulations). Furthermore, no jumps of size 2 are
observed in the plot. Thus, overestimating the obtained running time 
by $1$ step would have
been sufficient to complete gossiping in all of our test runs.

To gain empirical insights into the behavior of the memory-based approach
described in \autoref{sect:memory} under the assumption of node failures, we
implemented nodes that are marked as failed. These nodes
simply do not store any incoming message and refuse to transmit messages to
other nodes.

The plot in \autoref{plot:robustness} shows the results of simulations on an
Erd\H{o}s-R\'enyi random graph consisting of 1,000,000 nodes with an expected
node degree of $\log^2{n} \approx 400$. Our simulation of \autoref{alg:algorithm-3} constructed
$3$ message distribution trees, independently. Afterwards we marked $F$ nodes
chosen uniformly at random as failed. The nodes were deactivated
before Phase II. The $x$-axis in
\autoref{plot:robustness} shows this number of nodes $F$. In the
simulation, we determined the number of initial messages that have been lost in addition to
the messages of the $F$ marked nodes. \autoref{plot:robustness} shows on the $y$-axis the
ratio of the lost messages of healthy nodes over $F$.
That is, zero indicates that no additional initial message was lost, whereas
2.0 indicates that for every failed node the 
initial messages of at least two additional healthy nodes were not present in any tree root after
Phase II.

Further plots showing additional graph sizes and various levels of detail can
be found in \autoref{appendix:plots}.



\bibliographystyle{abbrv}

\bibliography{gossiping}

\newpage
\appendix

\section{Additional Lemmas from \cite{Els06}}
\label{appendix:lemmas}

For some $u,v$ let $A_{u,v}$ denote the event that $u$ and $v$ are 
connected by an edge, 
and let $A_{u,v,l}$ denote the event that $u$ and $v$ share an edge \textbf{and} $u$ chooses $v$ in 
step $l$ (according to the random phone call model). 
In the next lemma, we deal with the distribution of the neighbors of a 
node $u$ in a graph $G(n,p)$,
after it has chosen $t$ neighbors, uniformly at random, in 
$t=O(\log n)$ consecutive steps. In particular, we show that the probability 
of $u$ being connected with some node $v$, not chosen within these $t$ steps, 
is not substantially modified after 
$O(\log n)$ steps.
%

\begin{lemmax}[1 of \cite{Els06}]
\label{auxlem}
Let $V=\{v_1 , \dots , v_n\}$ be a set of $n$ nodes and let 
every pair of nodes $v_i,v_j$ be connected with probability 
$p$, independently, where $p \geq \log^{\delta} n/n$ for some constant $\delta >
 2$.
If $t = O(\log n)$, $u,v \in V$, and \[ A(U_0,U_1,U_2) = 
\bigwedge\limits_{{0 < l \leq t} \atop {(v_i,v_j,l) \in U_0}}
A_{v_{i},v_j,l}
\bigwedge\limits_{(v_{i'},v_{j'}) \in U_1} 
A_{v_{i'},v_{j'}} 
\bigwedge\limits_{(v_{i''},v_{j''}) 
\in U_2}
\overline{A_{(v_{i''},v_{j''})}} \enspace ,\] 
for some $U_0 \subset V \times V \times \{0, \dots , t\}$ and $U_1,U_2 \subset V
 \times V$,
then it holds that 
$$\Pr\left[ (u,v) \in E~\left|~
A(U_0,U_1,U_2) \right. \right] = p(1 \pm O(t/d)) ,$$
for any $U_0, U_1,U_2$ satisfying the following properties: 
\begin{itemize}
\item $|U_0 \cap \{ (v_i,v_j,l) | v_j \in V\} | = 1$ for any $v_i \in V$ and $l 
\in \{ 0, \dots , 
t\}$,
\item $|U_1 \cap \{ (u,u') | u' \in V\} | = \Omega (d)$ and  
$|U_1 \cap \{ (v,v') | v' \in V\} | = \Omega (d)$, 
\item $(u,v) \not\in U_1 \cup U_2$, and $(u,v,i) \not\in U_0$ for any $i$.  
\end{itemize}
\end{lemmax}

\section{Additional Lemmas from \cite{ES08}}

\begin{lemmax}[2.2 from \cite{ES08}, adapted version]\label{kssvlemma}
Let \autoref{alg:algorithm-3} be executed on the 
graph $G(n,p)$ 
of size $n$, where 
$p > \log^{2+\Omega(2)} n / n$
and $\rho$ is a properly chosen (large) constant. 
If $t=\log n + \frac{\rho}{2} \log \log n$, then $|H(t)| \leq n/\sqrt[4]{d}$
 and the 
number of 
transmissions after $t$ time steps is bounded by $O(n \log \log n)$. 
Additionally, if $t=\log n + \frac{3 \rho}{8} \log \log n$, we have $|H(t)| 
\geq n/\sqrt[4]{d}$. 
\end{lemmax}

\begin{lemmax}[2.7 from \cite{ES08}, adapted version]
\label{lem_5} Let $|H(t)| \in [\log^q n,n/\sqrt[4]{d}]$ be the number of
uninformed nodes in $G(n,p)$ at some time $t=O(\log n)$, where $q$ is a large 
constant, and let \autoref{alg:algorithm-3} be executed on this graph. 
Then, $|H(t
+3 \rho \log \log n/8)| \leq 
\log^q n$, w.h.p., provided that $\rho$ is large enough.
\end{lemmax}

\begin{lemmax}[2.8 from \cite{ES08}, adapted version]
\label{lem_6} Let $|H(t)| \leq \log^q n$ be the number of uninformed
nodes in $G(n,p)$ at time $t = O(\log n)$, and let
\autoref{alg:algorithm-3} be executed on this graph. 
Then 
within additional $\rho \log \log n/8$ steps all nodes in the 
graph will be informed, w.h.p., whenever $\rho$ is large enough. 
\end{lemmax}

\newpage
\section{Further Plots and Empirical Data}
\label{appendix:plots}
\subsection{Push-Pull-Algorithm}

The algorithm labeled \emph{push-pull} in \autoref{plot:comparison} is a simple
procedure, where in each step every node opens a connection, performs
\texttt{pushpull}, and closes the connection. For the sake of completeness of
this paper it is described in pseudocode in \autoref{alg:simple-push-pull}.

\begin{algorithm}
\SetAlgoVlined
\SetKwFunction{Push}{push}
\SetKwFunction{Pull}{pull}
\SetKwFunction{PushPull}{pushpull}
\SetKwFunction{Add}{add}
\SetKwFunction{Empty}{empty}
\SetKwFunction{Pop}{pop}
\SetKwFor{At}{at each}{do in parallel}{end}
\SetKwFor{ForEach}{for each}{do}{end}
\SetKwFor{With}{with probability}{do}{end}
\SetKwProg{AlgorithmPI}{Phase \RNum{1}}{}{end}
\SetKwProg{AlgorithmPII}{Phase \RNum{2}}{}{end}
\SetKwProg{AlgorithmPIII}{Phase \RNum{3}}{}{end}
\SetKw{False}{false}
\SetKw{True}{true}

\SetKwFor{AtNode}{at each node}{do in parallel}{end}
\SetKwFor{For}{for}{do}{end}
\SetKwFor{ForRound}{for round}{do}{end}
\For{$t = 1$ \KwTo $\BigO{\log{n}}$}{
	\At{node $v$}{
		\PushPull{$m_v$}\;
	}
}

\caption{A simple push-pull algorithm to perform randomized gossiping. The \texttt{pushpull} operation is preceeded and followed by opening and closing a channel, respectively.}
\label{alg:simple-push-pull}
\end{algorithm}

\subsection{Robustness Analysis}

The following plots shown in \autoref{plot:robust-100k-500k} were run on graphs
of size 100,000 and 500,000 nodes, respectively. They visualize the results of
the same type of simulation as presented in \autoref{plot:robustness}.

More detail can be obtained from the plots shown in
\autoref{plot:robust-dense-100k-500k}, where we ran our simulation with a
higher resolution. That is, we ran a series of at least $5$ tests per number of
failed nodes. The number of failed nodes was chosen from the set $\{0,
100, 200, 300, \dots \}$.  We used graphs of two different sizes in these 6
plots. The left column shows the results for a graph consisting of 100,000
nodes and the right column for 500,000 nodes. The $x$-axis shows the number of
failed nodes, the $y$-axis shows the percentage of runs in which more than
a certain number $T$ of additional nodes failed. This number is $T=0$ for the
top row, $T=10$ for the middle row and $T=100$ for the bottom row. For example,
this tells us that on a graph of size 100,000 more than 4000 nodes could fail
and still the number of additional uninformed nodes was less than 100 in all
test runs.

\begin{figure}[b]
\centering
\includegraphics[scale=1.5]{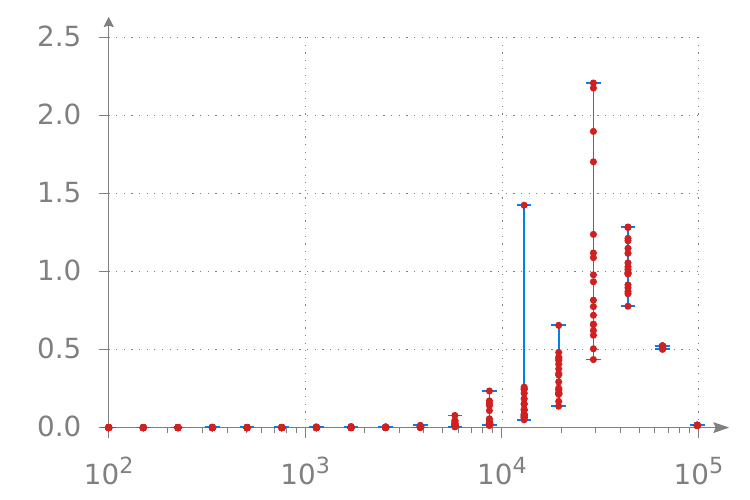}\\

\includegraphics[scale=1.5]{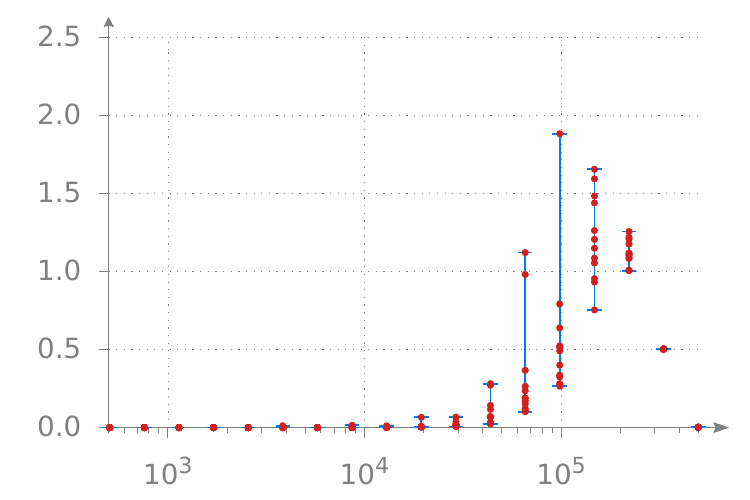}
\caption{Relative number of additional node failures in the memory model with
graphs of sizes 100,000 (top) and 500,000 (bottom). The $x$-axis shows the
number of failed nodes $F$, the $y$-axis the ratio of additional uninformed
nodes to $F$.}
\label{plot:robust-100k-500k}
\end{figure}

\begin{figure}[b]
\centering
\includegraphics[scale=1.5]{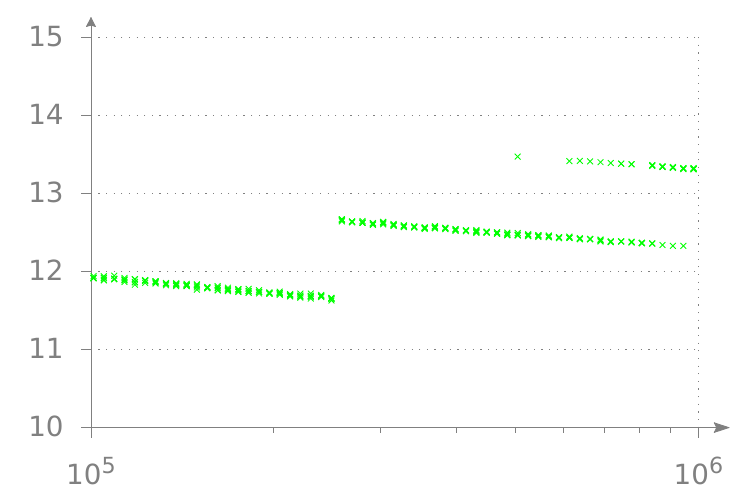}
\caption{A more detailed view of the data presented in \autoref{plot:comparison} for \autoref{alg:fast-gossiping}. The $x$-axis shows the graph size, the $y$-axis the number of messages sent per node.}
\label{plot:comparison2}
\end{figure}

\begin{figure*}[p]
\noindent\begin{minipage}[t]{0.475\textwidth}
\centering
\includegraphics[page=1]{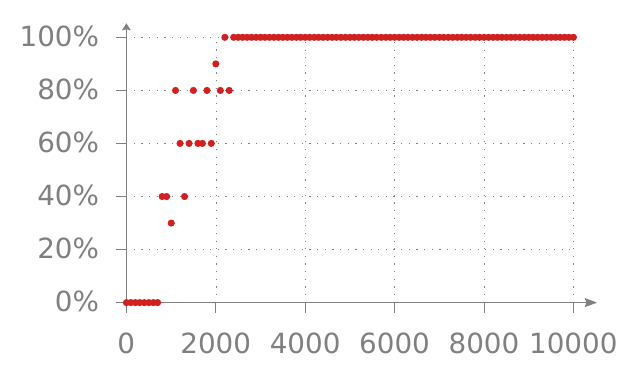}\\
\includegraphics[page=2]{plots/robust-dense-100k}\\
\includegraphics[page=3]{plots/robust-dense-100k}\\
\end{minipage}\hfill
\begin{minipage}[t]{0.475\textwidth}
\centering
\includegraphics[page=1]{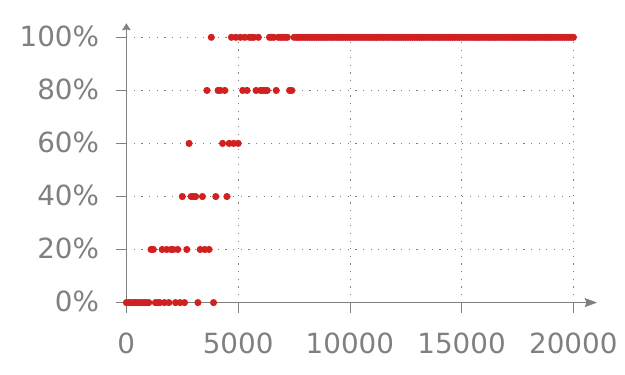}\\
\includegraphics[page=2]{plots/robust-dense-500k}\\
\includegraphics[page=3]{plots/robust-dense-500k}\\
\end{minipage}
\caption{Detailed plot showing the robustness of \autoref{alg:algorithm-3} on
graphs of sizes 100,000 (left column) and 500,000 (right column). The $x$-axis
shows the number of failed nodes, the $y$-axis shows the percentage of runs
in which more than $T$ additional numbers remained uninformed. In the top $T =
0$, in the middle $T = 10$ and in the bottom $T = 100$.}
\label{plot:robust-dense-100k-500k}
\vspace{\baselineskip}
\end{figure*}

\end{document}